\documentclass[article]{IEEEtran}

\usepackage{graphicx,colordvi,psfrag}
\usepackage{amsmath,amssymb}
\usepackage{epstopdf}
\usepackage[caption=false]{subfig}
\usepackage{epsfig,cite}
\usepackage{calc,pstricks, pgf, xcolor}
\usepackage{nicefrac}
\usepackage{enumerate}
\usepackage{bm}
\usepackage{fixltx2e}
\usepackage{pstricks-add}

\newtheorem{theorem}{Theorem}
\newtheorem{corollary}{Corollary}
\newtheorem{remark}{Remark}
\newtheorem{example}{Example}
\newtheorem{definition}{Definition}
\newtheorem{lemma}{Lemma}
\newtheorem{proposition}{Proposition}
\newenvironment{proof}[1][Proof]{\noindent\textbf{#1.} }{\ \rule{0.5em}{0.5em}}

\def\mymathhyphen{{\hbox{-}}}
\newcommand{\by}{\mathbf{y}}
\newcommand{\bY}{\mathbf{Y}}
\newcommand{\bx}{\mathbf{x}}
\newcommand{\bX}{\mathbf{X}}
\newcommand{\bz}{\mathbf{z}}
\newcommand{\bZ}{\mathbf{Z}}

\newcommand{\bH}{\mathbf{H}}
\newcommand{\bA}{\mathbf{A}}
\newcommand{\ba}{\mathbf{a}}

\newcommand{\bd}{\mathbf{d}}
\newcommand{\bD}{\mathbf{D}}
\newcommand{\bB}{\mathbf{B}}
\newcommand{\bb}{\mathbf{b}}
\newcommand{\bU}{\mathbf{U}}
\newcommand{\bu}{\mathbf{u}}
\newcommand{\bV}{\mathbf{V}}
\newcommand{\bv}{\mathbf{v}}
\newcommand{\CV}{\mathcal{V}}

\newcommand{\bt}{\mathbf{t}}
\newcommand{\bT}{\mathbf{T}}

\newcommand{\bI}{\mathbf{I}}
\newcommand{\bP}{\mathbf{P}}
\newcommand{\bL}{\mathbf{L}}

\newcommand{\bG}{\mathbf{G}}
\newcommand{\bQ}{\mathbf{Q}}
\newcommand{\bs}{\mathbf{s}}

\newcommand{\bF}{\mathbf{F}}

\newcommand{\ZZ}{\mathbb{Z}}
\newcommand{\CC}{\mathbb{C}}
\newcommand{\RR}{\mathbb{R}}

\newcommand{\Tsnr}{\mathsf{SNR}}

\newcommand{\Mod}{\bmod\Lambda}

\newcommand{\trace}{\mathop{\mathrm{trace}}}
\newcommand{\diag}{\mathop{\mathrm{diag}}}
\newcommand{\Span}{\mathop{\mathrm{span}}}

\newcommand{\vect}{\mathop{\mathrm{vec}}}
\newcommand{\Vol}{\mathrm{Vol}}

\newcommand{\Cwi}{C_{\text{WI}}}
\newcommand{\Cst}{\mathcal{C}^{\text{ST}}}
\newcommand{\MH}{\mathcal{H}}
\newcommand{\bZe}{\mathbf{Z}_{\text{eff}}}

\newcommand{\Real}{\mathsf{Re}}
\newcommand{\Imag}{\mathsf{Im}}

\DeclareMathOperator*{\argmin}{\arg\!\min}

\newrgbcolor{BoxRed}{1 .5 .5}
\newrgbcolor{BoxBlue}{.9 .9 1}
\newrgbcolor{lightgrey}{.7 .7 0.7}

\begin{document}

\title{Precoded Integer-Forcing Universally Achieves the MIMO Capacity to Within a Constant Gap}
\author{Or~Ordentlich and
        Uri~Erez,~\IEEEmembership{Member,~IEEE}
\thanks{The work of O. Ordentlich was supported by the Adams Fellowship Program of the Israel Academy of Sciences and Humanities, a fellowship from The Yitzhak and Chaya Weinstein Research Institute for Signal Processing at Tel Aviv University and the Feder Family Award. The work of U. Erez was supported in part by the Israel Science Foundation under Grant No. 1557/13.}
\thanks{O. Ordentlich and U. Erez are with Tel Aviv University, Tel Aviv, Israel (email: ordent,uri@eng.tau.ac.il).
}}


\maketitle

\begin{abstract}
An open-loop single-user multiple-input multiple-output communication scheme is considered where a transmitter, equipped with multiple antennas, encodes the data into independent streams all taken from the same linear code. The coded streams are then linearly precoded using the encoding matrix of a perfect linear dispersion space-time code.
At the receiver side, integer-forcing equalization is applied, followed by standard single-stream decoding.
It is shown that this communication architecture achieves the capacity of any Gaussian multiple-input multiple-output channel up to a gap that depends only on the number of transmit antennas.
\end{abstract}

\section{Introduction}
The Gaussian Multiple-Input Multiple-Output (MIMO) channel has been the focus of extensive research efforts since the pioneering works of Foschini \cite{foschini96}, Foschini and Gans \cite{fmg98}, and Telatar \cite{telatar99}.
Mathematically, the single-user complex MIMO channel with $M$ transmit and $N$ receive antennas is modeled as
\begin{align}
\by=\bH\bx+\bz\label{channelmodel}
\end{align}
where $\bH\in\CC^{N\times M}$ is the channel matrix, $\by\in\CC^{N\times 1}$ is the channel output, $\bx\in\CC^{M\times 1}$ is the input vector that is subject to the power constraint\footnote{In this paper $(\bx)^{\dagger}$ is the conjugate transpose of $\bx$.}
\begin{align}
\mathbb{E}(\bx^{\dagger}\bx)\leq M\cdot\Tsnr,\nonumber
\end{align}
and $\bz$ is an additive noise vector of i.i.d. circularly symmetric complex Gaussian entries with zero mean and unit variance.

The mutual information of this channel is maximized by a circularly symmetric complex Gaussian input~\cite{telatar99} with covariance matrix $\bQ$ satisfying $\trace(\bQ)\leq M\cdot \Tsnr$, and is given by\footnote{All logarithms in this paper are to base $2$, and rates are measured in bits per channel use.}
\begin{align}
C=\max_{\bQ\succ 0 \ : \ \trace{\bQ}\leq M\cdot\Tsnr}\log\det\left(\bI+\bQ\bH^{\dagger}\bH\right).\label{Capacity}
\end{align}
The choice of $\bQ$ that maximizes~\eqref{Capacity} is determined by the water-filling solution.
When the matrix $\bH$ is known at both transmission ends, i.e., in a closed-loop scenario, this mutual information is the capacity of the channel and may closely be approached using the singular-value decomposition in conjunction with standard scalar codes designed for
an additive white Gaussian noise (AWGN) channel.
In certain scenarios, the natural choice \mbox{$\bQ=\Tsnr\cdot\bI$} is used, resulting in the \emph{white-input (WI) mutual information}
\begin{align}
\log\det\left(\bI+\Tsnr\bH^{\dagger}\bH\right).\nonumber
\end{align}
We may define the set
\begin{align}
\mathbb{H}(\Cwi,\Tsnr)=\bigg\{&\bH\in\CC^{N\times M} \ : \nonumber\\
&\ \log\det\left(\bI+\Tsnr\bH^{\dagger}\bH\right)=\Cwi \bigg\},\label{compoundfamily}
\end{align}
of all channel matrices with the same white-input mutual information $\Cwi$. The corresponding \emph{compound channel} model is defined by~\eqref{channelmodel} with the channel matrix $\bH$ \emph{arbitrarily} chosen from the set $\mathbb{H}(\Cwi,\Tsnr)$, and fixed throughout the whole transmission period. The matrix $\bH$ that was chosen by the channel is revealed to the receiver, but not to the transmitter. Clearly, the capacity of this compound channel is $\Cwi$, and is achieved with a white Gaussian input. This paper is concerned with approaching the compound capacity using a low-complexity scheme.

The compound MIMO channel model appears in several important communication scenarios. Wireless systems often operate in \emph{open-loop} mode, where the receiver knows the channel matrix $\bH$ but the transmitter only knows the corresponding white-input mutual information. This scenario is well captured by the compound model, and will be the focus of this paper. One may be even more conservative in the assumptions on the channel state information available at the transmitter (CSIT), and assume that even $\Cwi$ is unknown. In this case, a reasonable approach is to transmit codewords from an i.i.d. white Gaussian codebook with target rate $R$, such that the receiver will be able to correctly decode the transmitted message if $R<\log\det\left(\bI+\Tsnr\bH^{\dagger}\bH\right)$. It follows that from the transmitter's perspective, the coding task for this scenario is identical to that of coding for a compound channel with $\Cwi=R$. It may be argued that if the channel matrix $\bH$ remains constant for a long period, the receiver can communicate (a quantized version of) it to the transmitter with a negligible overhead, which reduces the communication problem to the simpler closed-loop scenario. Sometimes, however, the transmitter wishes to broadcast the same message to many receivers, such that all receivers with a ``good-enough'' link should be able to decode the information. This communication model approaches the compound channel model~\eqref{compoundfamily} as the number of potential receivers grows.


While the theoretical performance limits of open-loop communication over a Gaussian MIMO channel are well understood, unlike for closed-loop transmission, much is still lacking when it comes to practical schemes that are able to approach these limits. In general, the notion of practicality is rather vague and can be understood in different ways. In this paper, we use it in the following sense: a scheme is deemed practical if it \emph{decouples} the signal-processing task of channel equalization from the coding task. In other words, a practical scheme applies simple signal processing operations to transform the MIMO channel to a set of scalar channels, over which standard ``off-the-shelf'' codes for an AWGN channel may be used. This notion of practicality is motivated by the fact that in the past decades, coding for AWGN channels has reached an advanced state, and low-complexity coding schemes (e.g., turbo and LDPC codes) operating near capacity are known. It is thus desirable to combine AWGN coding and decoding techniques with equalization in a modular way, with the aim of approaching the capacity of the MIMO channel.  For the closed-loop scenario, this can be achieved using the singular-value decomposition. However, for the compound MIMO channel, practical capacity-approaching schemes are not known in general.

Such a modular scheme is known for the $1 \times 2$ MISO channel where Alamouti modulation offers an optimal solution. More generally,
modulation via orthogonal space-time block ``codes" allows one to approach the WI mutual information using scalar AWGN coding and decoding
in the limit of small rate~\cite{tseviswanath}.

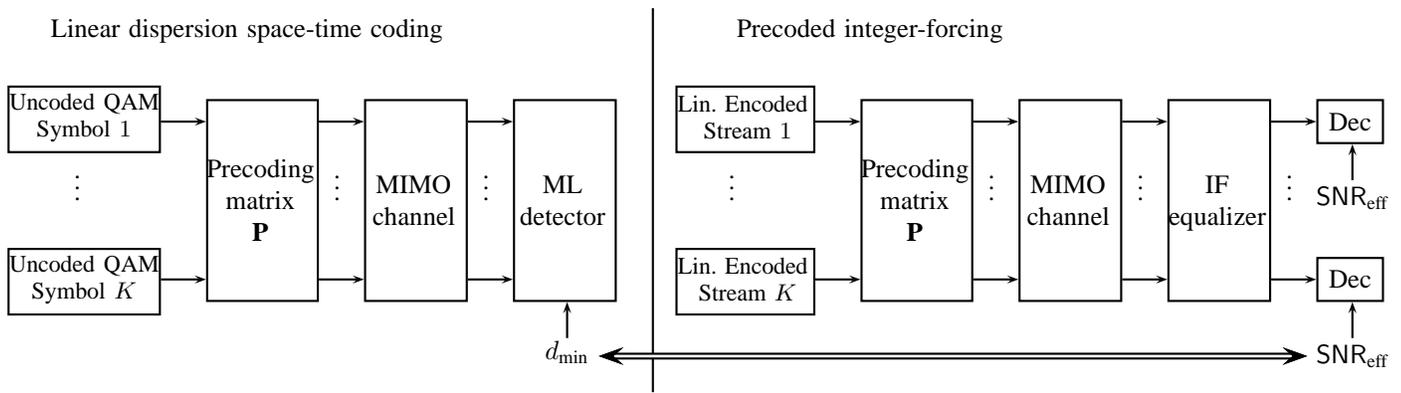
\begin{figure*}[]
\begin{center}
\psset{unit=0.6mm}
\begin{pspicture*}(-11,-20)(300,65)

\rput(0,0)
{
\rput(45,60){Linear dispersion space-time coding}

\rput(3,3){\psframe(-11,30)(23,45)\rput(6,41){\small{Uncoded QAM}}\rput(6,35){\small{Symbol $1$}}
}
\rput(7.5,27){$\vdots$}
\rput(3,-33){\psframe(-11,30)(23,45)\rput(6,41){\small{Uncoded QAM}}\rput(6,35){\small{Symbol $K$}}
}
\rput(26,0){
\rput(0,40){\psline{->}(0,0)(10,0)}
\rput(0,5){\psline{->}(0,0)(10,0)}
}

\rput(36,0){\psframe(0,0)(25,45)\rput(12,22){$\begin{array}{c}
                                             $\text{Precoding}$ \\
                                             $\text{matrix}$ \\
                                             $\textbf{P}$
                                           \end{array}$}
}

\rput(61,0){
\rput(0,40){\psline{->}(0,0)(10,0)}
\rput(4,27){$\vdots$}
\rput(0,5){\psline{->}(0,0)(10,0)}
}

\rput(71,0){\psframe(0,0)(23,45)\rput(11,22){$\begin{array}{c}
                                             $\text{MIMO}$ \\
                                             $\text{channel}$
                                           \end{array}$}
}

\rput(94,0){
\rput(0,40){\psline{->}(0,0)(10,0)}
\rput(4,27){$\vdots$}
\rput(0,5){\psline{->}(0,0)(10,0)}
}

\rput(104,0){\psframe(0,0)(23,45)\rput(11,22){$\begin{array}{c}
                                             $\text{ML}$ \\
                                             $\text{detector}$
                                               \end{array}$}
\rput(12,-11){$d_{\text{min}}$}
\rput(12,0){\psline{->}(0,-8)(0,0)}
}

\rput(135,0){\psline(0,-20)(0,80)}

}

\rput(145,0)
{
\rput(38,60){Precoded integer-forcing}

\rput(3,3){\psframe(-8,30)(23,45)\rput(7,41){\small{Lin. Encoded}}\rput(8,35){\small{Stream $1$}}
}
\rput(7.5,27){$\vdots$}
\rput(3,-33){\psframe(-8,30)(23,45)\rput(7,41){\small{Lin. Encoded}}\rput(8,35){\small{Stream $K$}}
}

\rput(26,0){
\rput(0,40){\psline{->}(0,0)(10,0)}
\rput(0,5){\psline{->}(0,0)(10,0)}
}

\rput(36,0){\psframe(0,0)(25,45)\rput(12,22){$\begin{array}{c}
                                             $\text{Precoding}$ \\
                                             $\text{matrix}$ \\
                                             $\textbf{P}$
                                           \end{array}$}
}

\rput(61,0){
\rput(0,40){\psline{->}(0,0)(10,0)}
\rput(4,27){$\vdots$}
\rput(0,5){\psline{->}(0,0)(10,0)}
}

\rput(71,0){\psframe(0,0)(23,45)\rput(11,22){$\begin{array}{c}
                                             $\text{MIMO}$ \\
                                             $\text{channel}$
                                           \end{array}$}
}

\rput(94,0){
\rput(0,40){\psline{->}(0,0)(10,0)}
\rput(4,27){$\vdots$}
\rput(0,5){\psline{->}(0,0)(10,0)}
}

\rput(104,0){\psframe(0,0)(23,45)\rput(11,22){$\begin{array}{c}
                                             $\text{IF}$ \\
                                             $\text{equalizer}$
                                               \end{array}$}

}

\rput(127,0){
\rput(0,40){\psline{->}(0,0)(10,0)}
\rput(4,27){$\vdots$}
\rput(0,5){\psline{->}(0,0)(10,0)}
}

\rput(137,0){
\rput(0,40){\psframe(0,-5)(15,5)\rput(7.5,0){Dec}
\rput(8,-17){$\Tsnr_{\text{eff}}$}
\rput(8,0){\psline{->}(0,-13)(0,-5)}
}

\rput(0,5){\psframe(0,-5)(15,5)\rput(7.5,0){Dec}
\rput(8,-17){$\Tsnr_{\text{eff}}$}
\rput(8,0){\psline{->}(0,-13)(0,-5)}
}
}

}

\psline[doubleline=true]{<->}(123,-12)(280,-12)

\end{pspicture*}
\end{center}
\caption{An illustrative comparison between linear dispersion space-time coding and precoded integer-forcing. Linear dispersion space-time coding consists of precoding uncoded QAM symbols, and detecting these symbols at the receiver. The detector's performance is dictated by $d_{\text{min}}$ which is the minimum distance at the received constellation. In precoded integer-forcing, coded streams are precoded and transmitted over the channel. The receiver first applies an integer-forcing equalizer and then decodes linear combinations of the streams. The performance is dictated by $\Tsnr_{\text{eff}}$. In this paper we show that $d_{\text{min}}$ and $\Tsnr_{\text{eff}}$ are closely related.}
\label{fig:STvsPIF}
\end{figure*} 

Beyond the low rate regime, the multiple degrees of freedom offered by the channel need to be utilized in order to approach capacity. For this reason, despite considerable work and progress, the problem of designing a practical scheme that approaches the capacity of the compound MIMO channel remains unsolved. As a consequence, less demanding benchmarks became widely accepted in the literature. First,
since statistical modeling of a wireless communication link is often available, one may be content with guaranteeing good performance
only for channel realizations that have a ``high" probability. Further, to simplify analysis and design, the asymptotic criterion of the diversity-multiplexing tradeoff (DMT) \cite{zt03} has broadly been adopted.

Unfortunately, statistical characterizations, and the DMT criterion in particular, offer only a coarse figure of merit for assessing schemes. Specifically, assuming an i.i.d. fading model with a continuous distribution on the channel coefficients precludes the possibility of having an entire row in the channel matrix nulled out. For example, if the channel is assumed to have $N=2$ receive antennas and $M=2$ transmit antennas with  i.i.d. Rayleigh fading, the class of matrices of the form
\begin{align}
\bH=\left[
      \begin{array}{cc}
        h_1 & h_2 \\
        0 & 0 \\
      \end{array}
    \right]
    \label{nulledRow}
\end{align}
where $h_1$ and $h_2$ satisfy $\log(1+\Tsnr(|h_1|^2+|h_2|^2))=\Cwi$, has zero probability. Thus, the DMT optimality of a scheme w.r.t. a $2\times 2$ i.i.d. Rayleigh fading distribution, tells us nothing about its performance over channels of the form~\eqref{nulledRow}.
The class of channels described by~\eqref{nulledRow} corresponds to receivers that are equipped with a single antenna, rather than two. It follows that, a scheme that is DMT optimal for a $2\times 2$ i.i.d. Rayleigh fading distribution, may exhibit terrible performance over channels with dimensions $1\times 2$. Thus, the DMT framework is inadequate for analyzing communication scenarios with degrees-of-freedom mismatch, i.e., when the transmitter does not know in advance the number of receive antennas, or alternatively, has to simultaneously transmit (the same message) to several users, equipped with a different number of receive antennas. The compound channel model, on the other hand, does not distinguish between channel matrices with the same WI mutual information, and is therefore more suitable for such scenarios.

In~\cite{tv06}, Tavildar and Vishwanath introduced the notion of \emph{approximately universal space-time codes} and derived a necessary and sufficient criterion for a code to be approximately universal. This criterion is closely related to the nonvanishing determinant criterion and is met by several known coding schemes \cite{ekpkh06,orbv06,esk07}. Roughly speaking, approximate-universality guarantees that a scheme is DMT optimal for any statistical channel model. The criterion derived in~\cite{tv06} ensures that the minimum distance at the receiver scales appropriately with $\Cwi$ regardless of the exact realization of $\bH$, which, in turn, guarantees DMT optimality. Thus, the problem of finding coding schemes that are DMT optimal regardless of the channel statistics is now solved.


Approximately universal schemes still suffer, however, from the asymptotic nature of the DMT criterion.  Essentially, the approximate universality of a scheme guarantees that if the white-input mutual information of the MIMO channel is $\Cwi$, the scheme's error probability at a certain rate $R$ scales roughly as\footnote{The $Q$-function is defined as $Q(x)\triangleq\frac{1}{2\pi}\int_x^{\infty}e^{-\frac{t^2}{2}}dt$.} $Q(\sqrt{2^{\Cwi-R}})$, for large $\Cwi$. This is the same error probability behavior as that of uncoded transmission over a single-input single-output (SISO) AWGN channel with capacity $\Cwi$. This may suffice when $\Cwi$ is large enough and moderate error probabilities are required, but does not provide performance guarantees for finite values of $\Cwi$. In particular, the approximate universality criterion was designed for coding schemes with short block lengths, and does not attempt to exploit the opportunity of reducing the error probability by increasing the block length when the channel remains constant for a long period of time.

While designing a practical communication scheme that approaches the compound MIMO capacity is still out of reach, in the present work we take a step in this direction. Namely, a practical communication architecture that achieves the compound MIMO capacity up to a \emph{constant gap}, that depends only on the number of transmit antennas, is studied. Such a traditional information-theoretic performance guarantee is  substantially stronger than approximate universality.
In the considered scheme, which we refer to as \emph{precoded integer-forcing}, the transmitter encodes the data into independent streams, as in the standard V-BLAST~\cite{tseviswanath} architecture. However, in contrast to standard V-BLAST where each one of the streams can be encoded by a different code, in the considered scheme it is crucial that all streams are encoded using the \emph{same} linear code. The coded streams are then linearly precoded using the generating matrix of a space-time code from the class of perfect codes~\cite{yw03,brv05,orbv06,esk07}, which are approximately universal. At the receiver side, integer-forcing (IF) equalization~\cite{zneg12IT} is applied.

An IF receiver~\cite{zneg12IT} attempts to decode a full-rank set of linear combinations of the transmitted streams with integer-valued coefficients. Once these equations are decoded, they can be solved for the transmitted streams. The receiver's front end consists of a linear equalization matrix that transforms the MIMO channel into a set of SISO sub-channels, each corresponding to a different linear combination, with an effective SNR that depends on the integer coefficients of this linear combination. The performance of the scheme is dictated by the worst effective SNR, over all sub-channels.



\subsection{Our Contribution}
\label{subsec:ourcont}
The integer-forcing receiver architecture was introduced in~\cite{zneg12IT} and has since received considerable attention in the literature (see e.g.,~\cite{hc12,hc13,ncnc13,shv13}). While numerical experiments revealed that in many cases its performance is quite close to that of the optimal maximum-likelihood decoder~\cite{zneg12IT,shv13,de12}, the analytic performance guarantees available in the literature prior to this work were quite weak. In particular, the strongest result was that for $M\leq N$ the IF receiver achieves the optimal DMT for Rayleigh fading MIMO channels when the transmit antennas are restricted to transmitting independent streams~\cite{zneg12IT}. The main contribution of the current work is in providing solid analytic performance guarantees for the integer-forcing receiver.

The key step in our analysis is Lemma~\ref{lem:snrdmin} which lower bounds the effective SNR seen by the integer-forcing receiver in terms of $d_{\text{min}}$ - the minimum distance seen at the receiver when all antennas transmit QAM symbols. When the number of transmit antennas $M$ is larger than the number of receive antennas $N$, the minimum distance typically decreases as the cardinality of the QAM constellation increases. Our result, takes this phenomena into account and is therefore useful for any number of transmit and receive antennas. We then apply Lemma~\ref{lem:snrdmin} together with a recent result from number theory that concerns the typical rate of decrease of $d_{\text{min}}$ with the cardinality of the transmitted constellation~\cite{hussain2009metrical} to prove Lemma~\ref{lem:IF_DoF} which establishes that the IF receiver achieves the optimal number of degrees-of-freedom (DoF) for almost all $\bH\in\RR^{N\times M}$, regardless of $N$ and $M$. While this result is not surprising for the case $N\geq M$, where standard zero-forcing or MMSE receivers suffice to achieve the maximal number of DoF, it is quite remarkable for channels with $M<N$, where standard linear receivers are practically useless in the high-SNR regime.

Although Lemma~\ref{lem:IF_DoF} provides strong motivation for using the IF receiver, it suffers from two shortcomings. First, it characterizes the performance of the IF receiver only in the asymptotic high-SNR regime. Second, it only holds for almost all $\bH\in\RR^{N\times M}$ w.r.t. Lebesgue measure on $\RR^{N\times M}$, but provides no guarantees for specific channel realizations. To circumvent these weaknesses, we employ space-time precoding at the transmitter, resulting in a \emph{precoded IF} scheme.

Precoded IF may be viewed as an extension of linear dispersion space-time ``codes''. In such ``codes'', uncoded QAM symbols are linearly modulated over space and time. This is done by linearly precoding the QAM symbols using a precoding matrix $\bP$. For precoded IF, the same precoding matrix $\bP$ is applied to \emph{codewords} taken from a linear code, rather than uncoded QAM symbols. See Figure~\ref{fig:STvsPIF}. The performance of linear dispersion space-time ``codes'' is dictated by $d_{\text{min}}$, the minimum distance in the received constellation, whereas the performance of precoded IF is determined by the effective signal-to-noise ratio $\Tsnr_{\text{eff}}$. By Lemma~\ref{lem:snrdmin}, minimum distance guarantees for precoded QAM symbols translate to guarantees on the effective SNR for precoded IF, when the same precoding matrix is used.

The design of precoding matrices for uncoded QAM, that guarantee an appropriate growth of $d_{\text{min}}$ as a function of $\Cwi$, has been extensively studied over the last decade. A remarkable family of such matrices are the generating matrices of \emph{perfect} linear dispersion space-time codes, which are approximately universal~\cite{orbv06,esk07}. We apply the tight connection between $d_{\text{min}}$ and $\Tsnr_{\text{eff}}$ to show that when such precoding matrices are used for precoded IF, $\Tsnr_{\text{eff}}$ also grows appropriately with $\Cwi$. Consequently, we are able to prove that precoded IF achieves rates within a constant gap from the compound MIMO capacity.



\subsection{Related Work}
\label{subsec:relwork}
Integer-forcing equalization essentially reduces to lattice-reduction (LR) in the case of uncoded transmission. Lattice-reduction aided receivers for perfect space-time modulated QAM constellations were considered in the literature, and were shown to be DMT optimal~\cite{je10}. The key difference is that while the latter approach involves uncoded transmission and symbol-by-symbol detection, the architecture proposed here uses linearly \emph{coded} streams and the detection phase is replaced with equalization and decoding. This in turn, leads to performance guarantees that are valid at any (fixed) transmission rate.

In~\cite{ecd04}, El-Gamal \emph{et al.} proposed a lattice space-time (LAST) coding scheme, and showed that it can achieve the compound MIMO capacity. Although the LAST coding scheme uses lattice encoding and decoding, its complexity is in general very high. The reason for this is that the lattice decoding performed by the receiver is w.r.t. a lattice induced by both the transmitted constellation and the channel matrix $\bH$. In other words, the LAST coding scheme does not decouple the equalization and decoding tasks. In particular, even if a lattice with low decoding complexity is transmitted, after passing through the channel its structure is changed and the decoding complexity of the obtained lattice may (and is most likely to) no longer be low. This is not the case for precoded IF. In the scheme considered here, the receiver decodes integer linear combinations of the transmitted streams. Since these streams are taken from the same linear code, their integer linear combinations are also members of the linear code. As a result, the task of decoding these linear combinations is identical to the task of decoding a single stream over a scalar AWGN channel. If the linear/lattice code that was used to encode the streams can be decoded with low complexity, so can the integer linear combinations. The channel matrix $\bH$ is handled in the equalization procedure, and has no effect on the decoding task, just as in standard linear receiver architectures.

Finding the exact capacity region of many network information theoretic problems may be very difficult. Nevertheless, a recent line of work has demonstrated that characterizing the capacity region to within a constant number of bits is often a manageable challenge (see e.g.,~\cite{etw08,adt11,od13,tse2009s} and references therein). The constant gap result presented here is of different spirit. The capacity of the compound MIMO channel considered here is known and may be achieved using random coding and maximum-likelihood decoding. Our results only show that the rate achieved by the sub-optimal scheme precoded IF, is a constant number of bits from the capacity. Nevertheless, the results derived in this paper may be useful in the future for obtaining approximate capacity characterizations for several network problems. More specifically, it is now recognized that lattice codes play a key role in characterizing the fundamental limits of certain communication networks, see e.g.~\cite{pzek11,bpt10,ng11IT,wnps10,ncl10,CoFTransformFull} and~\cite[Chapter 12]{ramibook}. A common feature of many of these lattice-based coding schemes is that, from the perspective of each receiver, they induce effective multiple-access (MAC) channels with a reduced number of users, all of which employ the same lattice codebook. The achievable rates for a MAC channel where all users use the same lattice codebook is difficult to analyze, but can be lower bounded by the rates attained via the IF receiver. In~\cite{CoFTransformFull} this technique was successfully applied for approximating the sum-capacity for the symmetric Gaussian $K$-user interference channel. Our bounds on the rate-loss incurred by the IF receiver w.r.t. the mutual information may lead to closed form inner bounds on the performance of lattice-based coding schemes for other networks.

\subsection{Paper Outline}
The rest of the paper is outlined as follows. Section~\ref{sec:IF} gives an overview of IF equalization and analyzes its performance without precoding under various assumptions, while Section~\ref{sec:IFST} considers the precoded IF scheme. In Section~\ref{sec:ST} several properties of perfect linear dispersion space-time codes are recalled and a lower bound on their worst-case minimum distance is derived. The proof that precoded IF achieves the compound MIMO capacity to within a constant gap is given in Section~\ref{sec:main}. As an example of the advantages of the proposed approach, low-complexity constructions of MIMO \emph{rateless} coding schemes, which are based on precoded IF, are derived in Section~\ref{sec:rateless}.
Concluding remarks appear in Section~\ref{sec:conclusions}.

\section{Performance of the Integer-Forcing scheme}
\label{sec:IF}

Integer-Forcing equalization is a low-complexity architecture for the MIMO channel, which was proposed by Zhan \textit{et al}.~\cite{zneg12IT}. The key idea underlying IF is to first decode integral linear combinations of the signals transmitted by all antennas, and then, after the noise is removed, invert those linear combinations to recover the individual transmitted signals. This is made possible by transmitting codewords from the \emph{same} linear\slash lattice code from all $M$ transmit antennas, leveraging the property that linear codes are closed under (modulo) linear combinations with integer-valued coefficients.

In this section we review and extend some of the results of~\cite{zneg12IT} and~\cite{de12} in a way that is suitable for our purposes.

\subsection{Nested Lattice Codes}
\label{subsec:nested}
Let $\Lambda_c\subset\Lambda_f$ be a pair of $n$-dimensional nested lattices (see~\cite{ez04,ramibook} for a more thorough treatment of lattice definitions and properties). The lattice $\Lambda_c$ is referred to as the coarse lattice and $\Lambda_f$ as the fine lattice. Denote by $\CV_c$ the fundamental Voronoi region of $\Lambda_c$, and define the second moment of $\Lambda_c$ as
\begin{align}
\sigma^2(\Lambda_c)\triangleq\frac{1}{n}\frac{1}{\Vol(\CV_c)}\int_{\bu\in\CV_c}\|\bu\|^2d\bu,\nonumber
\end{align}
where $\Vol(\CV_c)$ is the volume of $\CV_c$. A nested lattice codebook $\mathcal{C}=\Lambda_f\cap\CV_c$, with rate
\begin{align}
R=\frac{1}{n}\log\left|\Lambda_f\cap\CV_c\right|\frac{\text{bits}}{\text{channel use}}\nonumber
\end{align}
is associated with the nested lattice pair. The codebook is scaled such that $\sigma^2(\Lambda_c)=\Tsnr/2$.

\vspace{2mm}

\begin{example}
\label{ex:latticecodebooks}
We give three examples of common structures of nested lattice codebooks. See Figure~2 for an illustration. More examples can be found in~\cite{fsk11}.

\begin{itemize}
\item \underline{\emph{Uncoded transmission}} - The simplest nested lattice codebook is an uncoded one, where the fine lattice $\Lambda_f$ is the integer lattice $\ZZ$ whereas the coarse lattice is $\Lambda_c=q\ZZ$ for some integer $q>1$. The Voronoi region in this case is \mbox{$\CV_c=[-q/2,q/2)$} and the obtained nested lattice codebook $\mathcal{C}$ consists of all integers in the interval $[-q/2,q/2)$. The rate of this codebook is $R=\log{q} \ \text{bits}/\text{channel use}$.
\item \underline{\emph{$q$-ary linear code without shaping}} - A more sophisticated, yet reasonable to implement, nested lattice codebook can be obtained by lifting a $q$-ary linear code with block length $n$ to Euclidean space using Construction A~\cite{cs88,loeliger97}, and taking the resulting lattice as $\Lambda_f$. The coarse lattice is taken as $\Lambda_c=q\ZZ^n$, as in the uncoded case. The obtained nested lattice codebook $\mathcal{C}$ is therefore simply the $q$-ary linear code coupled with a PAM constellation.
\item \underline{\emph{``Good'' nested lattice pair of high dimension}} - A third option is to use a pair of  lattices of high dimension where the fine lattice is ``good'' for coding over an AWGN channel, whereas the coarse lattice is ``good'' for mean-squared-error quantization (see~\cite{ez04,ramibook} for precise definitions of ``goodness''). The obtained nested lattice codebook admits a relatively simple performance analysis, that yields closed-form rate expressions. However, implementing such a codebook is more complicated (although progress in this direction was made in~\cite{et05}).

    The performance improvement obtained by using such a codebook w.r.t. a $q$-ary linear code without shaping is bounded from above by $\nicefrac{1}{2}\log(2\pi e/12)$ bits per real dimension, provided that the $q$-ary linear code performs well over an AWGN channel.
\end{itemize}
\end{example}

\begin{figure*}[t]
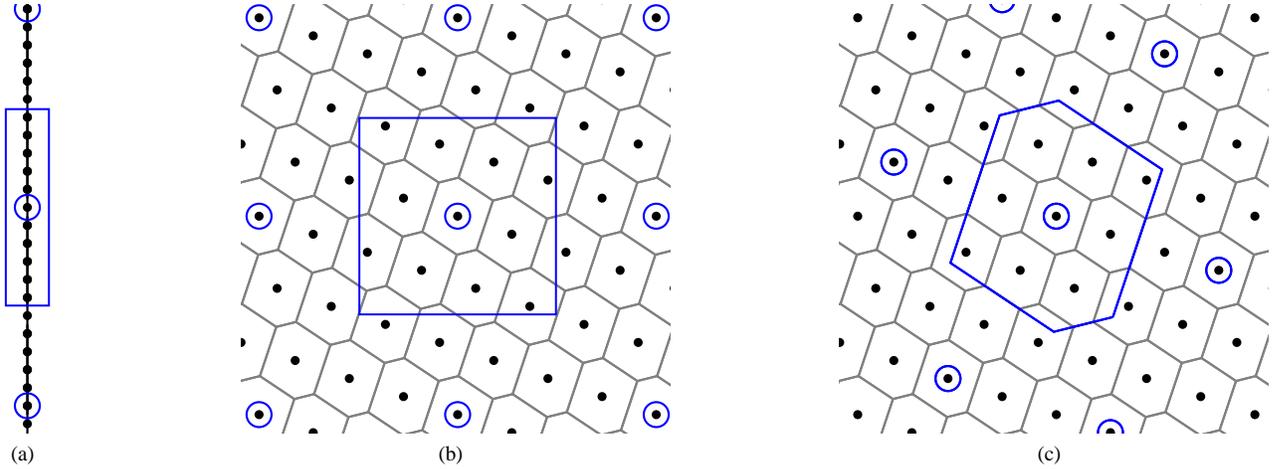

\label{fig:nestedlattices}
\subfloat[]{
\begin{minipage}[]{0.15\linewidth}
\centering
\psset{unit=0.6mm}
\begin{center}
\begin{pspicture*}(-10,-50)(10,45)

\rput(0,0){
\rput(0,0){\input{lattice1Dq11}}
\rput(0,-24){\input{lattice1Dq11}}
\rput(0,24){\input{lattice1Dq11}}
\rput(0,-48){\input{lattice1Dq11}}
\rput(0,48){\input{lattice1Dq11}}
\rput(0,0){\pscircle[linecolor=blue](0,0){3}\psframe[linecolor=blue,linewidth=0.75pt](-5,-22)(5,22)}
\rput(0,44){\pscircle[linecolor=blue](0,0){3}}
\rput(0,-44){\pscircle[linecolor=blue](0,0){3}}
}
\end{pspicture*}
\end{center} 
\end{minipage}}
\hspace{0.5cm}
\subfloat[]{
\begin{minipage}[]{0.4\linewidth}
\centering
\psset{unit=.60mm}
\begin{center}
\begin{pspicture*}(40,40)(135,135)

\rput(0,0){
\rput(0,0){\input{lattice2_3mod11}{\pscircle[linecolor=blue](0,0){3}}}
\rput(0,44){\input{lattice2_3mod11}{\pscircle[linecolor=blue](0,0){3}}}
\rput(44,0){\input{lattice2_3mod11}{\pscircle[linecolor=blue](0,0){3}}}
\rput(44,88){\input{lattice2_3mod11}{\pscircle[linecolor=blue](0,0){3}}}
\rput(44,132){\input{lattice2_3mod11}{\pscircle[linecolor=blue](0,0){3}}}
\rput(88,0){\input{lattice2_3mod11}{\pscircle[linecolor=blue](0,0){3}}}
\rput(0,88){\input{lattice2_3mod11}{\pscircle[linecolor=blue](0,0){3}}}
\rput(88,88){\input{lattice2_3mod11}{\pscircle[linecolor=blue](0,0){3}}}
\rput(132,88){\input{lattice2_3mod11}{\pscircle[linecolor=blue](0,0){3}}}
\rput(88,132){\input{lattice2_3mod11}{\pscircle[linecolor=blue](0,0){3}}}
\rput(132,44){\input{lattice2_3mod11}{\pscircle[linecolor=blue](0,0){3}}}
\rput(88,44){\input{lattice2_3mod11}{\pscircle[linecolor=blue](0,0){3}}}
\rput(44,44){\input{lattice2_3mod11}{\pscircle[linecolor=blue](0,0){3}}}
\rput(132,132){\input{lattice2_3mod11}{\pscircle[linecolor=blue](0,0){3}}}
\rput(66,66){\psframe[linecolor=blue,linewidth=0.75pt](0,0)(44,44)}
}

\end{pspicture*}
\end{center} 
\end{minipage}}
\hspace{0.5cm}
\subfloat[]{
\begin{minipage}[]{0.4\linewidth}
\centering
\psset{unit=.60mm}
\begin{center}
\begin{pspicture*}(40,40)(135,135)

\rput(0,0){
\rput(0,0){\input{lattice2_3mod11}}
\rput(0,44){\input{lattice2_3mod11}}
\rput(44,0){\input{lattice2_3mod11}}
\rput(44,88){\input{lattice2_3mod11}}
\rput(44,132){\input{lattice2_3mod11}}
\rput(88,0){\input{lattice2_3mod11}}
\rput(0,88){\input{lattice2_3mod11}}
\rput(88,88){\input{lattice2_3mod11}}
\rput(132,88){\input{lattice2_3mod11}}
\rput(88,132){\input{lattice2_3mod11}}
\rput(132,44){\input{lattice2_3mod11}}
\rput(88,44){\input{lattice2_3mod11}}
\rput(44,44){\input{lattice2_3mod11}}
\rput(132,132){\input{lattice2_3mod11}}

\rput(88,88){\input{lattice2_3mod11x3}}
\rput(88,88){\input{lattice2_3mod11x3}}
}

\end{pspicture*}
\end{center} 
\end{minipage}}
\caption{An illustration of the three different types of nested lattice codebooks given in Example~\ref{ex:latticecodebooks}. In all three cases the black points correspond to the fine lattice points, the blue circles to the coarse lattice points, and the blue polygon corresponds to the shaping region. In (a) the constellation for uncoded transmission with $q=11$ is illustrated. In (b) a $q$-ary linear code without shaping is shown, with $q=11$. In (c) a ``good'' nested lattice pair in two-dimensions is illustrated.}
\end{figure*}

\subsection{Description of the IF scheme}
\label{subsec:IFdescription}

In the IF scheme, the information bits to be transmitted are partitioned into $2M$ streams, labeled $\{1_{\Real},1_{\Imag},\ldots,M_{\Real},M_{\Imag}\}$. Each of the $2M$ streams is encoded by the nested lattice code $\mathcal{C}$, producing $2M$ row vectors, each in $\mathcal{C}\subset\RR^{1\times n}$. In particular, the stream $m_{\Real}$, consisting of $nR$ information bits, is mapped to a lattice point $\bt_{m_{\Real}}\in\mathcal{C}$. Then, a random dither $\bd_{m_{\Real}}\in\RR^{1\times n}$ uniformly distributed over $\CV_c$ and statistically independent of $\bt_{m_{\Real}}$, known to both the transmitter and the receiver, is used to produce the signal
\begin{align}
\bx_{m_{\Real}}=\left[\bt_{m_{\Real}}-\bd_{m_{\Real}}\right]\Mod_c\nonumber.
\end{align}
The signal $\bx_{m_{\Real}}$ is uniformly distributed over $\CV_c$ and is statistically independent of $\bt_{m_{\Real}}$ due to the Crypto Lemma~\cite[Lemma 1]{ez04}. It follows that
\begin{align}
\frac{1}{n}\mathbb{E}\|\bx_{m_{\Real}}\|^2=\sigma^2(\Lambda_c)=\frac{\Tsnr}{2}.\nonumber
\end{align}
A similar procedure is used to construct the signal $\bx_{m_{\Imag}}$.
The $m$th antenna transmits the signal $\bx_m=\bx_{m_{\Real}}+i\bx_{m_{\Imag}}\in\CC^{1\times n}$ over $n$ consecutive channel uses. Thus, the total transmission rate is $R_{\text{IF}}=2MR$ bits/channel use.

Let $\bX\triangleq[\bx_1^T \ \cdots \ \bx_M^T]^T\in\CC^{M\times n}$. The received signal is
\begin{align}
\bY=\bH\bX+\bZ,\nonumber
\end{align}
where $\bZ\in\CC^{N\times n}$ is a vector with i.i.d. circularly symmetric complex Gaussian entries. Letting the subscripts $\Real$ and $\Imag$ denote the real and imaginary parts of a matrix, respectively, the channel can be expressed by its real-valued representation
\begin{align}
\left[
  \begin{array}{c}
    \bY_{\Real} \\
    \bY_{\Imag} \\
  \end{array}
\right]
=\left[
   \begin{array}{cc}
     \bH_{\Real} & -\bH_{\Imag} \\
     \bH_{\Imag} & \bH_{\Real} \\
   \end{array}
 \right]\left[
  \begin{array}{c}
    \bX_{\Real} \\
    \bX_{\Imag} \\
  \end{array}
\right]+\left[
  \begin{array}{c}
    \bZ_{\Real} \\
    \bZ_{\Imag} \\
  \end{array}
\right],\label{realrep}
\end{align}
which will be written as $$\tilde{\bY}=\tilde{\bH}\tilde{\bX}+\tilde{\bZ}$$ for notational compactness. Let $$\tilde{\bT}\triangleq[\bt^T_{1_{\Real}} \ \cdots \ \bt^T_{M_{\Real}} \ \bt^T_{1_{\Imag}} \ \cdots \ \bt^T_{M_{\Imag}}]^T$$ be a $2M\times n$ real-valued matrix whose rows consist of the lattice points corresponding to the $2M$ bit streams, and $$\tilde{\bD}\triangleq[\bd^T_{1_{\Real}} \ \cdots \ \bd^T_{M_{\Real}} \ \bd^T_{1_{\Imag}} \ \cdots \ \bd^T_{M_{\Imag}}]^T$$ be a $2M\times n$ real-valued matrix whose rows correspond to the $2M$ different dither vectors.

\begin{figure*}[]
\begin{center}
\psset{unit=0.6mm}
\begin{pspicture}(0,-10)(250,80)

\rput(0,78){Transmitter}
\psframe[linestyle=dashed](-18,-5)(47,75)

\rput(68,78){Channel}
\psframe[linestyle=dashed](56.5,-5)(111,75)

\rput(130,78){Receiver}
\psframe[linestyle=dashed](114,-5)(250,75)

\rput(0,50){
\rput(0,0){info. bits 1} \psline{->}(15,0)(25,0) \psframe(25,-5)(45,5)
\rput(35,0){Enc}\psline{->}(45,0)(60,0) \rput(53,3.5){$\bx_1$}
}

\rput(0,30){\rput(35,2){$\vdots$}}

\rput(0,10){
\rput(0,0){info. bits M} \psline{->}(15,0)(25,0) \psframe(25,-5)(45,5)
\rput(35,0){Enc}\psline{->}(45,0)(60,0) \rput(53,3.5){$\bx_M$}
}

\rput(60,0){
\psframe(0,0)(25,60)\rput(12.5,30){$\bH$}}

\rput(85,0){
\rput(0,55){
\psline{->}(0,0)(12,0)
\pscircle(15,0){3}
\psline(12,0)(18,0)
\psline(15,-2)(15,2)
\rput(18,0){\psline{->}(0,0)(15,0)}
\rput(22,3){$\by_1$}
\psline{->}(15,8)(15,3)
\rput(15,12){$\bz_1$}
}

\rput(0,30){\rput(15,2){$\vdots$}}

\rput(0,5){
\psline{->}(0,0)(12,0)
\pscircle(15,0){3}
\psline(12,0)(18,0)
\psline(15,-2)(15,2)
\rput(18,0){\psline{->}(0,0)(15,0)}
\rput(22,3){$\by_N$}
\psline{->}(15,8)(15,3)
\rput(15,12){$\bz_N$}
}
}

\rput(118,0){
\psframe(0,0)(25,60)\rput(12.5,30){$\bB$}}

\rput(143,0){
\rput(0,50){
\psline{->}(0,0)(10,0) \psframe(10,-5)(30,5)
\rput(20,0){Dec}\psline{->}(30,0)(45,0) \rput(38,3.5){$\hat{\bv}_1$}
}

\rput(0,30){\rput(21.5,2){$\vdots$}}

\rput(0,10){
\psline{->}(0,0)(10,0) \psframe(10,-5)(30,5)
\rput(20,0){Dec}\psline{->}(30,0)(45,0) \rput(38,3.5){$\hat{\bv}_M$}
}
}

\rput(188,0){
\psframe(0,0)(25,60)\rput(12.5,30){$\bA^{-1}$}}

\rput(213,0){
\rput(0,50){
\psline{->}(0,0)(30,0)
\rput(15,5){$\widehat{\text{info. bits 1}}$}
}

\rput(0,30){\rput(15,2){$\vdots$}}

\rput(0,10){
\psline{->}(0,0)(30,0)
\rput(15,5){$\widehat{\text{info. bits M}}$}}
}
\end{pspicture}
\end{center}
\caption{A schematic overview of the integer-forcing transmitter and receiver. For simplicity, the dithers are not depicted in the figure, and a real-valued channel is assumed. At the transmitter, the information bits are split to $M$ streams. Each stream is encoded by the same linear codebook and transmitted by one of the transmit antennas. The receiver first applies the equalizing matrix $\bB$ whose role is to equalize the channel $\bH$ to an equivalent channel with transfer matrix approximately equal to $\bA$. The equalizer produces $M$ outputs, each of which is an integer-valued linear combination of the transmitted codewords plus effective noise. Each one of these outputs is decoded separately, and finally the outputs of the $M$ decoders are multiplied by $\bA^{-1}$ to produce the transmitted codewords. The codewords are then mapped to information bits (this step is not depicted in the figure).}
\label{fig:IFarch}
\end{figure*}
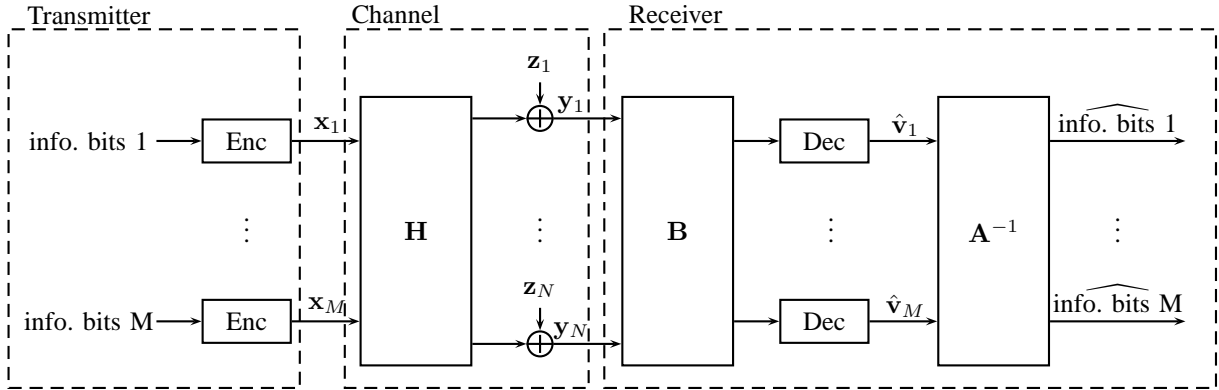 
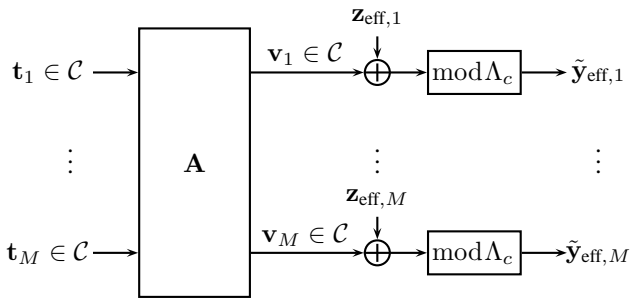
\begin{figure}[]
\begin{center}
\psset{unit=0.6mm}
\begin{pspicture}(0,0)(125,70)

\rput(0,50){
\rput(0,0){$\bt_1\in\mathcal{C}$} \psline{->}(10,0)(20,0)
}

\rput(0,30){\rput(5,2){$\vdots$}}

\rput(0,10){\rput(0,0){$\bt_M\in\mathcal{C}$} \psline{->}(10,0)(20,0)
}

\rput(20,0){
\psframe(0,0)(25,60)\rput(12.5,30){$\bA$}}

\rput(45,0){
\rput(0,50){
\psline{->}(0,0)(25,0)
\rput(12,4){$\bv_1\in\mathcal{C}$}
\pscircle(28,0){3}
\psline(25,0)(31,0)
\psline(28,-2)(28,2)
\rput(31,0){\psline{->}(0,0)(8,0)}
\psline{->}(28,8)(28,3)
\rput(28,12){$\bz_{\text{eff},1}$}
\rput(39,0){\psframe(0,-5)(21,5)\rput(10,0){$\Mod_c$}}
\rput(60,0){\psline{->}(0,0)(10,0)}
\rput(77,0){$\tilde{\by}_{\text{eff},1}$}
}

\rput(0,30){\rput(28,2){$\vdots$}\rput(77,2){$\vdots$}}

\rput(0,10){
\psline{->}(0,0)(25,0)
\rput(12,4){$\bv_M\in\mathcal{C}$}
\pscircle(28,0){3}
\psline(25,0)(31,0)
\psline(28,-2)(28,2)
\rput(31,0){\psline{->}(0,0)(8,0)}
\psline{->}(28,8)(28,3)
\rput(28,12){$\bz_{\text{eff},M}$}
\rput(39,0){\psframe(0,-5)(21,5)\rput(10,0){$\Mod_c$}}
\rput(60,0){\psline{->}(0,0)(10,0)}
\rput(77,0){$\tilde{\by}_{\text{eff},M}$}
}
}
\end{pspicture}
\end{center}
\caption{An illustration of the effective channel obtained when integer-forcing equalization is used. The effective channel consists of $M$ parallel sub-channels. The output of each sub-channel is an integer-valued linear combination of lattice points, which is itself a lattice point, plus effective noise, modulo the coarse lattice $\Lambda_c$.}
\label{fig:IFeffchannel}
\end{figure} 

The IF receiver chooses an equalizing matrix $\bB\in\RR^{2M\times 2N}$ and a full-rank target integer-valued matrix $\bA\in\ZZ^{2M\times 2M}$, and computes
\begin{align}
\tilde{\bY}_{\text{eff}}&=\left[\bB\tilde{\bY}+\bA\tilde\bD\right]\Mod_c\nonumber\\
&=\left[\bA\tilde{\bX}+\bA\tilde{\bD}+(\bB\tilde{\bH}-\bA)\tilde{\bX}+\bB\tilde{\bZ}\right]\Mod_c\nonumber\\
&=\left[\bA\tilde{\bT}+(\bB\tilde{\bH}-\bA)\tilde{\bX}+\bB\tilde{\bZ}\right]\Mod_c\nonumber\\
&=\left[\bV+\bZe\right]\Mod_c,\label{eqoutput}
\end{align}
where
\begin{align}
\bV\triangleq\left[\bA\tilde{\bT}\right]\Mod_c
\label{modeq}
\end{align}
is a $2M\times n$ real-valued matrix with each row being a codeword in $\mathcal{C}$ owing to the linearity of the code, $$\bZe\triangleq (\bB\tilde{\bH}-\bA)\tilde{\bX}+\bB\tilde{\bZ}$$ is additive noise statistically independent of $\bV$ (as $\tilde{\bX}$, as well as $\tilde{\bZ}$ are statistically independent of $\tilde{\bT}$), and the notation $\mod \Lambda_c$ is to be understood as reducing \emph{each row} of the obtained matrix modulo the coarse lattice.
Each row of $\tilde{\bY}_{\text{eff}}$ is the modulo sum of a codeword and effective noise. Thus, the IF receiver transforms the original MIMO channel into a set of $2M$ point-to-point modulo-additive sub-channels
\begin{align}
\tilde{\by}_{\text{eff},k}=\left[\bv_k+\bz_{\text{eff},k}\right]\Mod_c, \ \ k=1,\ldots,2M.\label{effsubchannel}
\end{align}
The additive noise vectors $\bz_{\text{eff},1},\ldots,\bz_{\text{eff},2M}$ are not statistically independent. Therefore, strictly speaking, the $2M$ effective channels $\tilde{\by}_{\text{eff},1},\ldots,\tilde{\by}_{\text{eff},2M}$ are not parallel. However, the IF decoder ignores the correlation between the noise vectors and decodes the output of each sub-channel separately.\footnote{Some improvement can be obtained by exploiting these correlations~\cite{oen13,znoeg10}. Yet, we do not pursue this possibility in the present paper.} If decoding is successful over all $2M$ sub-channels, the receiver has access to $\bV$, from which it can recover the matrix $\tilde{\bT}$ by solving the (modulo) set\footnote{In~\cite{ozeng11} it is shown that it suffices that $\bA$ is invertible over $\RR$ in order to recover $\tilde{\bT}$ from $\bV$.} of equations~\eqref{modeq}. See Figures~\ref{fig:IFarch} and~\ref{fig:IFeffchannel}.

\vspace{2mm}

Let $\ba_k^T$ and $\bb_k^T$ be the $k$th rows of $\bA$ and $\bB$, respectively, and define the effective variance of $\bz_{\text{eff},k}$ as
\begin{align}
\sigma^2_{\text{eff},k}&\triangleq\frac{1}{n}\mathbb{E}\left\|\bz_{\text{eff},k}\right\|^2\nonumber\\
&=\frac{1}{n}\mathbb{E}\left\|(\bb_k^T\tilde{\bH}-\ba_k^T)\tilde{\bX}+\bb_k^T\tilde{\bZ}\right\|^2\nonumber\\
&=\frac{\Tsnr}{2}\|(\bb_k^T\tilde{\bH}-\ba_k^T)\|^2+\frac{1}{2}\|\bb_k^T\|^2.\nonumber
\end{align}
A natural criterion for choosing the equalizing matrix $\bB$ and the target integer-valued matrix $\bA$ is to minimize the effective noise variances. It turns out~\cite{zneg12IT} that for a given matrix $\bA$, the optimal choice of $\bB$ under this criterion is
\begin{align}
\bB^{\text{opt}}=\bA\tilde{\bH}^T\left(\frac{1}{\Tsnr}\bI+\tilde{\bH}\tilde{\bH}^T\right)^{-1}.\label{wiener}
\end{align}
The matrix in~\eqref{wiener} can be interpreted as first applying the linear MMSE estimator of $\tilde{\bX}$ form $\tilde{\bY}$, and then multiplying the result by the integer-valued matrix $\bA$. In general, the estimation errors after linear MMSE estimation may be highly correlated, and have different powers. The role $\bA$ plays here is in decreasing these correlations and balancing the power of the remaining estimation errors. The freedom to choose any full-rank $\bA\in\ZZ^{2M\times 2M}$ and not just $\bA=\bI$ comes from the fact that any integer-linear combination of codewords is a codeword itself.
Setting $\bB$ as in~\eqref{wiener} results in the effective variances
\begin{align}
\sigma^2_{\text{eff},k}=\frac{\Tsnr}{2}\ba_k^T\left(\bI+\Tsnr\tilde{\bH}^T\tilde{\bH}\right)^{-1}\ba_k,\nonumber
\end{align}
for $k=1,\ldots,2M$.

Define the effective signal-to-noise ratio (SNR) at the $k$th sub-channel as
\begin{align}
\Tsnr_{\text{eff},k}&\triangleq\frac{\sigma^2(\Lambda_c)}{\sigma^2_{\text{eff},k}}\nonumber\\
&=\frac{\frac{\Tsnr}{2}}{\frac{\Tsnr}{2}\ba_k^T\left(\bI+\Tsnr\tilde{\bH}^T\tilde{\bH}\right)^{-1}\ba_k}\nonumber\\
&=\left(\ba_k^T\left(\bI+\Tsnr\tilde{\bH}^T\tilde{\bH}\right)^{-1}\ba_k\right)^{-1},\label{kSNReff}
\end{align}
and let
\begin{align}
\Tsnr_{\text{eff}}\triangleq\min_{k=1,\ldots,2M}\Tsnr_{\text{eff},k}.\label{SNReff}
\end{align}
For IF equalization to be successful, decoding over all $2M$ sub-channels should be correct. Therefore, the worst sub-channel constitutes a bottleneck. For this reason, the total performance of the receiver is dictated by $\Tsnr_{\text{eff}}$.

\subsection{Achievable rates for IF}
\label{subsec:performance}
When the codebook $\mathcal{C}$ is constructed from a good pair of nested lattices (see Example~\ref{ex:latticecodebooks}), the distribution of the effective noise at each sub-channel $k$, which is a linear combination of an AWGN and $2M$ dither vectors, approaches (with the code's block length) that of an AWGN with zero mean and variance $\sigma^2_{\text{eff},k}$~\cite{ng11IT}. Good nested lattice codebooks can achieve any rate satisfying
\begin{align}
R<\frac{1}{2}\log\left(\Tsnr_{\text{eff},k}\right)\label{rateAWGN}
\end{align}
over a $\bmod\hspace{0.5mm}\mymathhyphen\Lambda_c$ AWGN channel with signal-to-noise ratio $\Tsnr_{\text{eff},k}$~\cite{ez04,ng11IT}. Since $\bv_k$ is a codeword from a good nested lattice code and $\bz_{\text{eff},k}$ approaches an AWGN in distribution, $\bv_k$ can be decoded~\cite{zneg12IT,ng11IT} from $\tilde{\by}_{\text{eff},k}$ as long as the rate of the codebook $\mathcal{C}$ satisfies~\eqref{rateAWGN}. It follows that as long as
\begin{align}
R<\frac{1}{2}\log\left(\Tsnr_{\text{eff}}\right),\nonumber
\end{align}
all sub-channels $k=1,\ldots,2M$ can decode their linear combinations $\bv_k$ without error, and therefore IF equalization can achieve any rate satisfying
\begin{align}
R_{\text{IF}}&<2M\frac{1}{2}\log\left(\Tsnr_{\text{eff}}\right)\nonumber\\
&=M\log\left(\Tsnr_{\text{eff}}\right).\label{Rif}
\end{align}
As mentioned in Example~\ref{ex:latticecodebooks}, good nested lattice codebooks can be difficult to implement in practice. A more appealing alternative may be to use a $q$-ary linear code without shaping. In this case, the effective noise $\bz_{\text{eff},k}$ at each sub-channel is a linear combination of an AWGN and $2M$ random dithers uniformly distributed over the Voronoi region of a $1$-D integer lattice. This effective noise is i.i.d. (in contrast to the case where a higher-dimensional coarse lattice is used where $\bz_{\text{eff},k}$ has memory). It was shown in~\cite[Remark 3]{oe12} that, for a prime $q$ large enough, $q$-ary linear codes without shaping can achieve any rate satisfying
\begin{align}
R<\frac{1}{2}\log\left(\Tsnr_{\text{eff}}\right)-\frac{1}{2}\log\left(\frac{2\pi e}{12}\right)\nonumber
\end{align}
over a modulo channel with additive i.i.d. effective noise $\bz_{\text{eff},k}$.
Therefore, IF equalization using $q$-ary linear codes without shaping can achieve any rate satisfying
\begin{align}
R_{\text{IF,q-ary}}<M\log\left(\Tsnr_{\text{eff}}\right)-M\log\left(\frac{2\pi e}{12}\right).\label{Rifqary}
\end{align}
When a specific $q$-ary linear code (such as an LDPC code or a turbo code) is used, the achievable rate is further degraded by $2M$ times the code's gap-to-capacity at the target error probability.

Finally, consider the case of uncoded transmission. In this case, \mbox{$\Lambda_f=\gamma\ZZ$} and \mbox{$\Lambda_c=\gamma q\ZZ$}, where \mbox{$\gamma=\sqrt{12\Tsnr/q^2}$} is chosen so as to meet the power constraint, and $q>1$ is an integer (see Example~\ref{ex:latticecodebooks}). The performance of uncoded transmission with IF equalization followed by a simple slicer is characterized by the following lemma.

\vspace{1mm}

\begin{lemma}
\label{lem:IFuncoded}
The error probability of the IF receiver with uncoded transmission rate $R_{\text{IF}}$ is upper bounded by
\begin{align}
P_{e,\text{IF-uncoded}}&\leq 4M\exp\left\{-\frac{3}{2}2^{\frac{1}{M}\left(M\log(\Tsnr_{\text{eff}})-R_{\text{IF}}\right)} \right\}.\label{peuncoded}
\end{align}
\end{lemma}

\vspace{1mm}

\begin{proof}
See Appendix~\ref{app:uncoded}
\end{proof}

\vspace{1mm}

\begin{remark}
\label{rem:LR}
Integer-forcing equalization with uncoded transmission is quite similar to the extensively studied lattice-reduction-aided linear decoders framework~\cite{yw02,wf03,je10}. However, two subtle differences should be pointed out. First, under the framework of LR-aided linear decoding, the target integer matrix $\bA$ has to be unimodular, i.e., it has to satisfy $|\det(\bA)|=1$, whereas in IF equalization $\bA$ is only required to be full-rank. Second, the use of the dithers in IF equalization results in statistical independence between $v_k$ and $z_{\text{eff},k}$ at each of the $2M$ sub-channels. This allows for an exact rigorous analysis of the error probability, which is seemingly difficult under the LR framework.
\end{remark}

\subsection{Bounding the Effective SNR for an optimal choice of $\bA$}
\label{subsec:snreffbound}
In this subsection, we derive a lower bound on $\Tsnr_{\text{eff}}$, which will subsequently be used to lower bound the achievable rate of IF. Since the IF scheme is compatible with any choice of full-rank integer matrix $\bA\in\ZZ^{2M\times 2M}$, we would like to choose $\bA$ so as to maximize $\Tsnr_{\text{eff}}$. We denote the rate-maximizing target integer-valued matrix by $\bA^{\text{opt}}$. For the remainder of the paper $\Tsnr_{\text{eff}}$ refers to the effective SNR corresponding to the choice $\bA=\bA^{\text{opt}}$.

Using~\eqref{kSNReff} and~\eqref{SNReff}, this maximization criterion translates to
\begin{align}
\bA^{\text{opt}}&=\argmin_{\substack{{\bA\in\ZZ^{2M\times 2M}}\\ {\det(\bA)\neq 0}}} \ \max_{k=1\ldots,2M}\ba_k^T\left(\bI+\Tsnr\tilde{\bH}^T\tilde{\bH}\right)^{-1}\ba_k.\nonumber
\end{align}
The matrix \mbox{$\left(\bI+\Tsnr\tilde{\bH}^T\tilde{\bH}\right)^{-1}$} is symmetric and positive definite, and therefore it admits a Cholesky decomposition
\begin{align}
\left(\bI+\Tsnr\tilde{\bH}^T\tilde{\bH}\right)^{-1}=\bL \bL^T,\label{choldef}
\end{align}
where $\bL$ is a lower triangular matrix with strictly positive diagonal entries. With this notation the optimization criterion becomes
\begin{align}
\bA^{\text{opt}}=\argmin_{\substack{{\bA\in\ZZ^{2M\times 2M}}\\ {\det(\bA)\neq 0}}} \ \max_{k=1\ldots,2M}\|\bL^T\ba_k\|^2.\nonumber
\end{align}
Denote by $\Lambda(\bL^T)$ the $2M$ dimensional lattice spanned by the matrix $\bL^T$, i.e.,
\begin{align}
\Lambda(\bL^T)\triangleq\left\{\bL^T\ba \ : \ \ba\in\ZZ^{2M} \right\}.\nonumber
\end{align}
It follows that $\bA^{\text{opt}}$ should consist of the set of $2M$ linearly independent integer-valued vectors that result in the shortest set of linearly independent lattice vectors in $\Lambda(\bL^T)$.

\vspace{1mm}

\begin{definition}[Successive minima]
\label{def:sucmin}
Let $\Lambda(\mathbf{G})$ be a lattice spanned by the full-rank matrix $\mathbf{G} \in \mathbb{R}^{K \times K}$. For \mbox{$k=1,\ldots,K$}, we define the $k$th successive minimum as
\begin{align}
\lambda_k(\mathbf{G}) \triangleq \inf\left\{r \ : \ \dim\left(\Span\left(\Lambda(\mathbf{G})\bigcap \mathcal{B}(\mathbf{0},r)\right)\right)\geq k\right\}\nonumber
\end{align}
where $\mathcal{B}(\mathbf{0},r)=\left\{\bx\in\RR^{K} \ : \ \|\bx\|\leq r\right\}$ is the closed ball of radius $r$ around $\mathbf{0}$. In words, the $k$th successive minimum of a lattice is the minimal radius of a ball centered around $\mathbf{0}$ that contains $k$ linearly independent lattice points.
\end{definition}

\vspace{1mm}

With the above definition of successive minima, the effective signal-to-noise ratio, when the optimal integer-valued matrix $\bA^{\text{opt}}$ is used, can be written as
\begin{align}
\Tsnr_{\text{eff}}=\frac{1}{\lambda^2_{2M}(\mathbf{L}^T)}.\label{snrtmp}
\end{align}
Bounding the value of the $2M$th successive minimum of a lattice is seemingly difficult. Fortunately, a transference theorem by Banaszczyk~\cite{Banaszczyk93} relates the $2M$th successive minimum of a lattice to the first successive minimum of its dual lattice. Following the derivation from~\cite[Proof of Theorem 5]{zneg12IT}, we proceed to bound $\Tsnr_{\text{eff}}$ using this relation.

\vspace{2mm}

\begin{definition}[Dual lattice]
For a lattice $\Lambda(\bG)$ with a generating full-rank matrix $\bG\in\RR^{2M\times 2M}$ the dual lattice is defined by
\begin{align}
\Lambda^*(\bG)&\triangleq\Lambda\left((\bG^T)^{-1}\right)\nonumber\\
&=\left\{(\bG^T)^{-1}\ba \ : \ \ba\in\ZZ^{2M} \right\}.\nonumber
\end{align}
\label{def:dual}
\end{definition}

\vspace{2mm}

\begin{theorem}[Banaszczyk {\cite[Theorem 2.1]{Banaszczyk93}}]
Let $\Lambda(\bG)$ be a lattice with a full-rank generating matrix $\bG\in\RR^{K\times K}$ and let \mbox{$\Lambda^*(\bG)=\Lambda\left((\bG^T)^{-1}\right)$} be its dual lattice. The successive minima of $\Lambda(\bG)$ and $\Lambda^*(\bG)$ satisfy the following inequality
\begin{align}
\lambda_k \left(\bG\right)\lambda_{K-k+1}\left((\bG^T)^{-1}\right) < K, \ \ \ \forall k  = 1,2, \ldots,K.\nonumber
\end{align}
\label{thm:dualsuc}
\end{theorem}

\begin{proof}
See~\cite{Banaszczyk93}
\end{proof}

\vspace{2mm}

The following theorem gives a lower bound for $\Tsnr_{\text{eff}}$.

\vspace{2mm}

\begin{theorem}
\label{thm:SNReff}
Consider the complex MIMO channel \mbox{$\by=\bH\bx+\bz$} with $M$ transmit antennas and $N$ receive antennas, power constraint \mbox{$\mathbb{E}(\bx^{\dagger}\bx)\leq M\cdot\Tsnr$}, and additive noise $\bz$ with i.i.d. circularly symmetric complex Gaussian entries with zero mean and unit variance. The effective signal-to-noise ratio when integer-forcing equalization is applied is lower bounded by
\begin{align}
\Tsnr_{\text{eff}}>\frac{1}{4M^2}\min_{\ba\in\ZZ^{M}+i\ZZ^{M}\setminus\mathbf{0}}\ba^{\dagger}\left(\bI+\Tsnr{\bH}^{\dagger}{\bH}\right)\ba.\label{snrbound}
\end{align}
\end{theorem}

\begin{proof}
Let $\tilde{\bH}$ be the real-valued representation of the channel $\bH$, as in~\eqref{realrep}, and let $\bL$ and $\bL^T$ be as in~\eqref{choldef}. From~\eqref{snrtmp} we have
\begin{align}
\Tsnr_{\text{eff}}=\frac{1}{\lambda^2_{2M}(\mathbf{L}^T)}.\nonumber
\end{align}
The dual lattice of $\Lambda(\mathbf{L}^T)$ is $\Lambda(\bL^{-1})$. Thus, Theorem~\ref{thm:dualsuc} gives
\begin{align}
\frac{1}{\lambda^2_{2M}(\mathbf{L}^T)}>\frac{1}{(2M)^2}\lambda_1^2(\bL^{-1}),\nonumber
\end{align}
and therefore
\begin{align}
\Tsnr_{\text{eff}}&>\frac{1}{(2M)^2}\lambda_1^2(\bL^{-1})\nonumber\\
&=\frac{1}{4M^2}\min_{\ba\in\ZZ^{2M}\setminus\mathbf{0}}\|\bL^{-1}\ba\|^2\nonumber\\
&=\frac{1}{4M^2}\min_{\ba\in\ZZ^{2M}\setminus\mathbf{0}}\ba^T(\bL \bL^{T})^{-1}\ba\nonumber\\
&=\frac{1}{4M^2}\min_{\ba\in\ZZ^{2M}\setminus\mathbf{0}}\ba^T\left(\bI+\Tsnr\tilde{\bH}^T\tilde{\bH}\right)\ba.\label{realsnrbound}
\end{align}
where~\eqref{realsnrbound} follows from~\eqref{choldef}. Since the matrix \mbox{$\left(\bI+\Tsnr\tilde{{\bH}}^T\tilde{{\bH}}\right)\in\RR^{2M\times 2M}$} is the real-valued representation of the complex matrix \mbox{$\left(\bI+\Tsnr{{\bH}}^{\dagger}{{\bH}}\right)\in\CC^{M\times M}$},~\eqref{realsnrbound} can be written in complex form as~\eqref{snrbound}.
\end{proof}

\vspace{2mm}

\begin{remark}
It is worth mentioning that the bound~\eqref{snrbound} is tight up to a multiplicative factor of $4M^2$. Namely, it can be easily shown~\cite[VIII.5, Theorem VI]{cassels1997} that for a full-rank matrix $\bG\in\RR^{K\times K}$
\begin{align}
\lambda_K \left(\bG\right)\lambda_{1}\left((\bG^T)^{-1}\right) \geq 1.\nonumber
\end{align}
Now, repeating the same derivation as in the proof of Theorem~\ref{thm:SNReff} with $\bG=\bL^T$ gives
\begin{align}
\Tsnr_{\text{eff}}\leq\min_{\ba\in\ZZ^{M}+i\ZZ^{M}\setminus\mathbf{0}}\ba^{\dagger}\left(\bI+\Tsnr{\bH}^{\dagger}{\bH}\right)\ba.\nonumber
\end{align}

\end{remark}

\subsection{Relation between the effective SNR and the minimum distance for uncoded QAM}
\label{subsec:QAM}

A basic communication scheme for the MIMO channel is transmitting independent uncoded QAM symbols from each antenna. In this case, the error probability strongly depends on the \emph{minimum distance} at the receiver. For a positive integer $L$, we define
\begin{align}
d_{\text{min}}(\bH,L)\triangleq \min_{\ba\in \text{QAM}^M(L)\setminus\mathbf{0}}\|\bH\ba\|,\label{dmindef}
\end{align}
where
\begin{align}
\text{QAM}(L)&\triangleq \ \left\{-L,-L+1,\ldots,L-1,L\right\}\nonumber\\
&+i\left\{-L,-L+1,\ldots,L-1,L\right\},\label{QAML}
\end{align}
and $\text{QAM}^M(L)$ is an $M$-dimensional vector whose components all belong to $\text{QAM}(L)$.
Note that if $L$ is an even integer, $d_{\text{min}}(\bH,L)$ is the minimum distance at the receiver when each antenna transmits symbols from a $\text{QAM}(L/2)$ constellation. This is true since
\begin{align}
\min_{\substack{{\bx_1,\bx_2\in\text{QAM}^M(L/2)}\\{\bx_1\neq\bx_2}}}\|\bH\bx_1-\bH\bx_2\|= \min_{\bx\in \text{QAM}^M(L)\setminus\mathbf{0}}\|\bH\bx\|.\nonumber
\end{align}

\vspace{2mm}

In the IF scheme \emph{there is no assumption} that QAM symbols are transmitted. Rather, each antenna transmits codewords taken from a linear codebook. Nevertheless, we show that the performance of the IF receiver over the channel $\bH$ can be tightly related to those of a \emph{hypothetical} uncoded QAM system over the same channel. See Figure~\ref{fig:STvsPIF}. Namely, $\Tsnr_{\text{eff}}$ is closely related to $d_{\text{min}}(\bH,L)$. This relation is formalized in the next key lemma, which is a simple consequence of Theorem~\ref{thm:SNReff}.

\vspace{2mm}

\begin{lemma}[Relation between $\Tsnr_{\text{eff}}$ and $d_{\text{min}}$]
\label{lem:snrdmin}
Consider the complex MIMO channel \mbox{$\by=\bH\bx+\bz$} with $M$ transmit antennas and $N$ receive antennas, power constraint \mbox{$\mathbb{E}(\bx^{\dagger}\bx)\leq M\cdot\Tsnr$}, and additive noise $\bz$ with i.i.d. circularly symmetric complex Gaussian entries with zero mean and unit variance. The effective signal-to-noise ratio when integer-forcing equalization is applied is lower bounded by
\begin{align}
\Tsnr_{\text{eff}}>\frac{1}{4M^2}\min_{L=1,2,\ldots}\left(L^2+\Tsnr d_{\text{min}}^2(\bH,L) \right),\nonumber
\end{align}
where $d_{\text{min}}^2(\bH,L)$ is defined in~\eqref{dmindef}.
\end{lemma}

\vspace{2mm}

\begin{proof}
The bound from Theorem~\ref{thm:SNReff} can be written as
\begin{align}
\Tsnr_{\text{eff}}>\frac{1}{4M^2}\min_{\ba\in\ZZ^{M}+i\ZZ^{M}\setminus\mathbf{0}}\|\ba\|^2+\Tsnr\|\bH\ba\|^2.\label{tmpsnrbound}
\end{align}
Let
\begin{align}
\rho(\ba)\triangleq \max_{m=1,\ldots,M}\max\left(|a_{m_{\Real}}|,|a_{m_{\Imag}}|\right),\nonumber
\end{align}
i.e., $\rho(\ba)$ is the maximum absolute value of all real and imaginary components of $\ba$. With this notation,~\eqref{tmpsnrbound} is equivalent to
\begin{align}
\Tsnr_{\text{eff}}&>\frac{1}{4M^2}\min_{L=1,2,\ldots} \ \min_{\substack{ {\ba\in\ZZ^{M}+i\ZZ^{M}\setminus\mathbf{0}} \\ {\rho(\ba)=L} }}\|\ba\|^2+\Tsnr\|\bH\ba\|^2\nonumber\\
&\geq \frac{1}{4M^2}\min_{L=1,2,\ldots}\left(L^2+\Tsnr d_{\text{min}}^2(\bH,L)\right),\nonumber
\end{align}
as desired.
\end{proof}

\vspace{2mm}

\begin{remark}
In the transmission scheme described above each antenna transmits an independent stream. Therefore, the bounds from Theorem~\ref{thm:SNReff} and Lemma~\ref{lem:snrdmin} continue to hold true for multiple access (MAC) channels with $M$ users equipped with a single transmit antenna and a receiver equipped with $N$ receive antennas, where the gains from the $m$th transmit antenna to the receiver are given by the $m$th column of $\bH$ and each user is subject to the power constraint \mbox{$\mathbb{E}\left(|x_k|^2\right)\leq\Tsnr$}.
\label{rem:mac}
\end{remark}

\vspace{2mm}

\begin{remark}
For real-valued $N\times M$ MIMO channels \mbox{$\by=\bH\bx+\bz$} with power constraint \mbox{$\mathbb{E}(\bx^{T}\bx)\leq M\cdot\Tsnr$}, and \mbox{$\bz\sim\mathcal{N}(0,\bI)$} the bound from Theorem~\ref{thm:SNReff} becomes
\begin{align}
\Tsnr_{\text{eff}}>\frac{1}{M^2}\min_{\ba\in\ZZ^{M}\setminus\mathbf{0}}\ba^{T}\left(\bI+\Tsnr{\bH}^{T}{\bH}\right)\ba,\nonumber
\end{align}
and the bound from Lemma~\ref{lem:snrdmin} becomes
\begin{align}
\Tsnr_{\text{eff}}>\frac{1}{M^2}\min_{L=1,2,\ldots}\left(L^2+\Tsnr \tilde{d}_{\text{min}}^2(\bH,L) \right),\nonumber
\end{align}
where
\begin{align}
\tilde{d}_{\text{min}}(\bH,L)&\triangleq \min_{\ba\in \text{PAM}^M(L)\setminus\mathbf{0}}\|\bH\ba\|,\nonumber\\
\text{PAM}(L)&\triangleq \ \left\{-L,-L+1,\ldots,L-1,L\right\}.\nonumber
\end{align}
\label{rem:real}
\end{remark}

The bound from Lemma~\ref{lem:snrdmin} and its real-valued counterpart from Remark~\ref{rem:real} exhibit a Diophantine tradeoff, i.e., they depend on how small the norm $\|\bH\ba\|^2$ can be made as a function of the largest component in the integer-valued vector $\ba$. The typical behavior of this minimal norm, is the subject of several results in the metrical theory of Diophantine approximation, see e.g.~\cite{kemble2005groshev,hussain2009metrical,hussain2012metrical}. Using these results we derive the following lemma, which is proved in Appendix~\ref{app:DoFlemma}

\vspace{2mm}

\begin{lemma}[DoF of Integer-Forcing]
For almost all real-valued MIMO channels (w.r.t. Lebesgue measure), IF equalization achieves the optimal number of degrees-of-freedom (DoF), i.e.,
\begin{align}
\lim_{\Tsnr\rightarrow\infty}\frac{R_{\text{IF}}(\Tsnr)}{\nicefrac{1}{2}\log(\Tsnr)}&=M\lim_{\Tsnr\rightarrow\infty}\frac{\nicefrac{1}{2}\log(\Tsnr_{\text{eff}})}{\nicefrac{1}{2}\log(\Tsnr)}\nonumber\\
&=\min(M,N).\nonumber
\end{align}
\label{lem:IF_DoF}
\end{lemma}

\vspace{2mm}

Standard linear equalizers, such as the zero-forcing equalizer, or the MMSE equalizer, fail to achieve the optimal number of DoF when $N<M$ (In fact, when $N<M$, they achieve zero DoF). In light of this fact, our result that IF equalization achieves the full DoF is notable.
As discussed in Remark~\ref{rem:mac}, this result is also applicable for the MIMO-MAC channel. Thus, for almost every real-valued MIMO-MAC channel with $M$ users equipped with a single transmit antenna and a receiver equipped with $N$ receive antennas, each user can achieve $\min(M,N)/M$ DoF using IF equalization. This extends~\cite[Corollary 6]{CoFTransformFull}, which only covered the case of $N=1$.

\section{Precoded Integer-Forcing}
\label{sec:IFST}

The performance of IF equalization over Rayleigh fading channels was studied in~\cite{zneg12IT} and it was shown that when $N\geq M$ the IF equalizer achieves the optimal receive DMT (corresponding to transmission of independent streams from each antenna). However, in order to approach the compound MIMO capacity, transmitting independent streams from each antenna is not sufficient.

Clearly, there are instances of MIMO channels for which the lower bound~\eqref{snrbound} on $\Tsnr_{\text{eff}}$ does not increase with the WI mutual information. For example, consider a channel $\bH$ where one of the $NM$ entries equals $h$ whereas all other gains are zero. For such a channel \mbox{$\Cwi=\log(1+|h|^2\Tsnr)$}, yet $\Tsnr_{\text{eff}}=1$ (and the bound~\eqref{snrbound} only gives \mbox{$\Tsnr_{\text{eff}}>1/(4M^2)$}). Thus, it is evident that IF equalization alone can perform arbitrarily far from $\Cwi$.

This problem can be overcome by transmitting linear combinations of multiple streams from each antenna. More precisely, instead of transmitting $2M$ linearly coded streams, one from the in-phase component and one from the quadrature component of each antenna, over $n$ channel uses, $2MT$ linearly coded streams are precoded by a unitary matrix and transmitted over $nT$ channel uses.

Domanovitz \textit{et al}.~\cite{de12} proposed to combine IF equalization with linear precoding.
The idea is to transform the $N\times M$ complex MIMO channel~\eqref{channelmodel} into an aggregate \mbox{$NT\times MT$} complex MIMO channel and then apply IF equalization to the aggregate channel. The transformation is done using a unitary precoding matrix \mbox{$\bP\in\CC^{MT\times MT}$}. Specifically, let \mbox{$\bar{\bx}\in\CC^{MT\times 1}$} be the input vector to the aggregate channel. This vector is multiplied by $\bP$ to form the vector \mbox{$\bx=\bP\bar{\bx}\in\CC^{MT\times 1}$} which is transmitted over the channel~\eqref{channelmodel} during $T$ consecutive channel uses. Let
\begin{align}
\MH=\bI_T\otimes\bH=\left[
                      \begin{array}{cccc}
                        \bH & 0 & \cdots & 0 \\
                        0 & \bH & \cdots & 0 \\
                        \vdots & \vdots & \ddots & \vdots \\
                        0 & 0 & \cdots & \bH \\
                      \end{array}
                    \right],\label{hconc}
\end{align}
where $\otimes$ denotes the Kronecker product. The output of the aggregate channel is obtained by stacking $T$ consecutive outputs of the channel~\eqref{channelmodel} one below the other and is given by
\begin{align}
\bar{\by}&=\mathcal{H}\bP\bar{\bx}+\bar{\bz}\nonumber\\
&=\bar{\bH}\bar{\bx}+\bar{\bz},\label{uncodedST}
\end{align}
where \mbox{$\bar{\bH}\triangleq\mathcal{H}\bP=(\bI_T\otimes\bH)\bP\in\CC^{NT\times MT}$} is the aggregate channel matrix, and  \mbox{$\bar{\bz}\in\CC^{NT\times 1}$} is a vector of i.i.d. circularly symmetric complex Gaussian entries. See Figure~\ref{fig:precodedIF}.

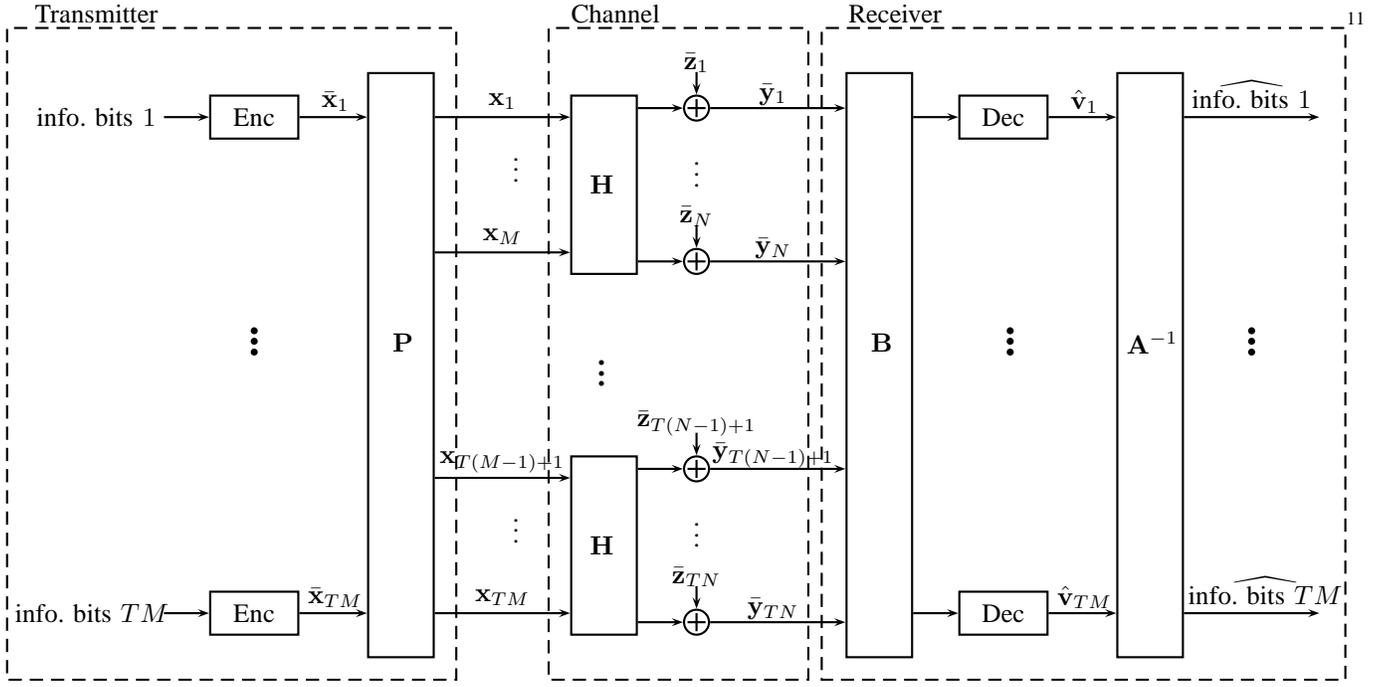
\begin{figure*}[]
\begin{center}
\psset{unit=0.6mm}
\begin{pspicture}(0,-10)(300,120)

\rput(20,0){

\rput(0,133){Transmitter}
\psframe[linestyle=dashed](-20,-15)(80,130)

\rput(115,133){Channel}
\psframe[linestyle=dashed](100,-15)(158,130)

\rput(177,133){Receiver}
\psframe[linestyle=dashed](160.5,-15)(277,130)

\rput(0,110){
\rput(0,0){info. bits $1$} \psline{->}(15,0)(25,0) \psframe(25,-5)(45,5)
\rput(35,0){Enc}\psline{->}(45,0)(60,0) \rput(53,3.5){$\bar{\bx}_1$}
}

\rput(0,60){\rput(35,2){\Huge$\vdots$}}

\rput(0,0){
\rput(-1    ,0){info. bits $TM$} \psline{->}(15,0)(25,0) \psframe(25,-5)(45,5)
\rput(35,0){Enc}\psline{->}(45,0)(60,0) \rput(53,3.5){$\bar{\bx}_{TM}$}
}

\rput(60,0){
\psframe(0,-10)(15,120)\rput(8,60){$\bP$}}

\rput(75,0){
\rput(0,110){
\rput(0,0){\psline{->}(0,0)(30,0) \rput(15,3.5){$\bx_{1}$}}
\rput(18,-10){$\vdots$}
\rput(0,-30){\psline{->}(0,0)(30,0) \rput(15,3.5){$\bx_{M}$}}
\rput(30,0){\psframe(0,-35)(15,5)\rput(7,-15){$\bH$}}
\rput(45,0){\rput(0,2){\psline{->}(0,0)(10,0)}\pscircle(13,2){3}
\psline(11,2)(15,2)
\psline(13,0)(13,4)
\psline{->}(13,10)(13,5)
\rput(13,12){$\bar{\bz}_1$}
\psline{->}(16,2)(46,2)\rput(30,5){$\bar{\by}_{1}$}
}

\rput(58,-11){$\vdots$}

\rput(45,-34){\rput(0,2){\psline{->}(0,0)(10,0)}\pscircle(13,2){3}
\psline(11,2)(15,2)
\psline(13,0)(13,4)
\psline{->}(13,10)(13,5)
\rput(13,12){$\bar{\bz}_N$}
\psline{->}(16,2)(46,2)\rput(30,5){$\bar{\by}_{N}$}
}
}

\rput(37,55){\LARGE$\vdots$}

\rput(0,30){
\rput(0,0){\psline{->}(0,0)(30,0) \rput(15,3.5){$\bx_{T(M-1)+1}$}}
\rput(18,-10){$\vdots$}
\rput(0,-30){\psline{->}(0,0)(30,0) \rput(15,3.5){$\bx_{TM}$}}
\rput(30,0){\psframe(0,-35)(15,5)\rput(7,-15){$\bH$}}
\rput(45,0){\rput(0,2){\psline{->}(0,0)(10,0)}\pscircle(13,2){3}
\psline(11,2)(15,2)
\psline(13,0)(13,4)
\psline{->}(13,10)(13,5)
\rput(13,12){$\bar{\bz}_{T(N-1)+1}$}
\psline{->}(16,2)(46,2)\rput(30,5){$\bar{\by}_{T(N-1)+1}$}
}

\rput(58,-11){$\vdots$}

\rput(45,-34){\rput(0,2){\psline{->}(0,0)(10,0)}\pscircle(13,2){3}
\psline(11,2)(15,2)
\psline(13,0)(13,4)
\psline{->}(13,10)(13,5)
\rput(13,12){$\bar{\bz}_{TN}$}
\psline{->}(16,2)(46,2)\rput(30,5){$\bar{\by}_{TN}$}
}
}

}

\rput(166,0){
\psframe(0,-10)(15,120)\rput(8,60){$\bB$}}

\rput(181,0){
\rput(0,110){
\psline{->}(0,0)(10,0) \psframe(10,-5)(30,5)
\rput(20,0){Dec}\psline{->}(30,0)(45,0) \rput(38,3.5){$\hat{\bv}_1$}
}

\rput(0,60){\rput(21.5,2){\Huge$\vdots$}}

\rput(0,0){
\psline{->}(0,0)(10,0) \psframe(10,-5)(30,5)
\rput(20,0){Dec}\psline{->}(30,0)(45,0) \rput(38,3.5){$\hat{\bv}_{TM}$}
}
}

\rput(226,0){
\psframe(0,-10)(15,120)\rput(8,60){$\bA^{-1}$}}

\rput(241,0){
\rput(0,110){
\psline{->}(0,0)(30,0)
\rput(15,5){$\widehat{\text{info. bits }1}$}
}

\rput(0,60){\rput(15,2){\Huge$\vdots$}}

\rput(0,0){
\psline{->}(0,0)(30,0)
\rput(18,5){$\widehat{\text{info. bits }TM}$}}
}
}
\end{pspicture}
\end{center}
\caption{A schematic overview of precoded integer-forcing. For simplicity, the dithers are not depicted in the figure, and a real-valued channel is assumed. At the transmitter, the information bits are split to $TM$ streams, each of which is encoded by the same linear code. Then, a $TM\times TM$ precoding matrix ``mixes'' the $TM$ codewords into $TM$ linear combinations. The channel $\bH$ is used $T$ times, where in each channel use each of the antennas transmits one of the precoded linear combinations. The receiver treats $T$ consecutive channel outputs as the output of an aggregate $NT\times MT$ channel with transfer matrix \mbox{$\bar{\bH}=(\bI_T\otimes\bH)\bP$}, and applies integer-forcing equalization to the aggregate channel.}
\label{fig:precodedIF}
\end{figure*} 

A remaining major challenge is how to choose the precoding matrix $\bP$ (recall that a compound channel is considered, and hence, the choice of $\bP$ cannot depend on $\bH$). As observed in Section~\ref{subsec:performance}, the performance of the IF equalizer is dictated by $\Tsnr_{\text{eff}}$. Thus, in order to obtain achievable rates that are comparable to the WI mutual information, $\Tsnr_{\text{eff}}$ must scale appropriately with $\Cwi$. The precoding matrix $\bP$ should therefore be chosen so as to guarantee this property for all channel matrices with the same WI mutual information.

Lemma~\ref{lem:snrdmin} indicates that for the aggregate channel $\Tsnr_{\text{eff}}$ is lower bounded by $\min_{L}(L^2+\Tsnr{d}_{\text{min}}(\bar{\bH},L))/4M^2$, where
\begin{align}
{d}_{\text{min}}(\bar{\bH},L)=\min_{\ba\in \text{QAM}^{MT}(L)\setminus\mathbf{0}}\|\mathcal{H}\bP\ba\|.\label{bardmindef}
\end{align}
Thus, the precoding matrix $\bP$ could be chosen so as to guarantee that ${d}_{\text{min}}(\bar{\bH},L)$ scales appropriately with $\Cwi$. This boils down to the problem of designing precoding matrices for transmitting QAM symbols over an unknown MIMO channel with the aim of maximizing the received minimum distance. The latter problem was extensively studied during the past decade, under the framework of linear dispersion space-time coding, and unitary precoding matrices that satisfy the aforementioned criterion were found. Therefore, the same matrices that proved so useful for space-time coding are also useful for precoded integer-forcing. A major difference, however, between the two is that while for linear dispersion space-time coding the precoding matrix $\bP$ is applied to uncoded QAM symbols, in precoded integer-forcing it is applied to \emph{coded streams}. This in turn, yields an achievable rate characterization for the compound MIMO channel which is not available using linear dispersion space-time coding. In particular, very different asymptotics can be analyzed. Rather than fixing the block length and taking $\Tsnr$ to infinity, as usually done in the space-time coding literature, here, we fix the channel and take the \emph{block length} to infinity, as in the traditional information-theoretic framework.

In~\cite{de12} the performance of IF equalization with the Golden code's~\cite{brv05} precoding matrix was numerically evaluated in a $2\times 2$ MIMO Rayleigh fading environment. The scheme's outage probability was found to be relatively close to that achieved by white i.i.d. Gaussian codebooks. Here, we prove that, in fact, precoded IF equalization, where the precoding matrix generates a perfect linear dispersion space-time code, achieves rates within a \emph{constant gap} from the compound MIMO capacity.

The aim of the next section is to lower bound ${d}_{\text{min}}(\bar{\bH},L)$ as a function of $\Cwi$ for precoding matrices $\bP$ that generate perfect linear dispersion space-time codes. This lower bound will be instrumental in proving that precoded IF universally attains the compound MIMO capacity to within a constant gap.
\section{Linear Dispersion Space-Time Codes}
\label{sec:ST}

Before deriving the lower bound on ${d}_{\text{min}}(\bar{\bH},L)$ some necessary background on space-time codes is given.

An $M\times T$ space-time (ST) code $\Cst$ for the channel~\eqref{channelmodel} with rate $R$ is a set of $|\Cst|=2^{RT}$ complex matrices of dimensions $M\times T$. The codebook $\Cst$ has to satisfy the average power constraint\footnote{The Frobenius norm of a matrix $\bX$ is denoted by $\|\bX\|_F^2$.}
\begin{align}
\frac{1}{2^{RT}}\sum_{\bX\in\Cst}\|\bX\|_F^2\leq MT\cdot\Tsnr.\nonumber
\end{align}
When the ST code $\Cst$ is used, a code matrix $\bX\in\Cst$ is transmitted column by column over $T$ consecutive channel uses, such that the $T$ channel outputs can be expressed as
\begin{align}
\bY=\bH\bX+\bZ, \nonumber
\end{align}
where each column of the matrices $\bY,\bZ\in\CC^{N\times T}$ represents the channel output and additive noise, respectively, at one of the $T$ channel uses.

An ST code $\Cst$ is said to be a \emph{linear dispersion} ST code~\cite{hh02} over the constellation $\mathcal{S}$ if every code matrix $\bX\in\Cst$ can be uniquely decomposed as
\begin{align}
\bX=\sum_{k=1}^K s_k \bF_k, \ s_k\in\mathcal{S},\nonumber
\end{align}
where $\mathcal{S}$ is some constellation and the matrices $\bF_k\in\CC^{M\times T}$ are fixed and independent of the constellation symbols $s_k$. Denoting by $\vect(\bX)$ the vector obtained by
stacking the columns of $\bX$ one below the other, and letting $\bs=[s_1 \ \cdots \ s_K]^T$ gives
\begin{align}
\vect(\bX)=\bP \bs,\nonumber
\end{align}
where
\begin{align}
\bP=[\vect(\bF_1) \ \vect(\bF_2) \ \cdots \ \vect(\bF_K)]\nonumber
\end{align}
is the code's $MT\times K$ \emph{generating matrix}. A linear dispersion ST code is \emph{full-rate} if $K=MT$.
In the sequel, linear dispersion ST codes over a $\text{QAM}(L)$ constellation, defined in~\eqref{QAML}, will play a key role. The linear dispersion ST code obtained by using the infinite constellation $\text{QAM}(\infty)=\ZZ+i\ZZ$ is referred to as $\Cst_{\infty}$, and, after vectorization, is in fact a complex lattice with generating matrix $\bP$. Since the $\text{QAM}(L)$ constellation is a subset of $\ZZ+i\ZZ$ it follows that for any finite $L$ the $\text{QAM}(L)$ based code $\Cst$ is a subset of $\Cst_{\infty}$.

An important class of linear dispersion ST codes with $T=M$ is that of \emph{perfect codes}~\cite{orbv06,esk07}, which is defined next.

\vspace{2mm}

\begin{definition}
An $M\times M$ linear dispersion ST code over a $\text{QAM}$ constellation is called \emph{perfect} if
\begin{enumerate}
\item It is full-rate;
\item It satisfies the nonvanishing determinant criterion
\begin{align}
\delta_{\text{min}}(\Cst_{\infty})\triangleq\inf_{\mathbf{0}\neq\bX\in\Cst_{\infty}}|\det(\bX)|^2>0;\nonumber
\end{align}
\item The code's generating matrix is unitary, i.e., \mbox{$\bP^{\dagger}\bP=\bI$}.\label{unitary}
\end{enumerate}
\label{def:perfectcodes}
\end{definition}

Note that this definition is slightly different than the one used in~\cite{orbv06,esk07}, where instead of condition~\ref{unitary} it is required that the energy of the codeword corresponding to the information symbols $\bs$ will have the same energy as $\|\bs\|^2$, and that all the coded symbols in all $T$ time slots will have the same average energy.

In~\cite{orbv06}, perfect linear dispersion ST codes were found for $M=2,3,4$ and $6$, whereas in~\cite{esk07} perfect linear dispersion ST codes were obtained for any positive integer $M$. The constructions in~\cite{orbv06,esk07} are based on cyclic division algebras, and result in unitary generating matrices. Thus, for any positive integer $M$, there exist codes that satisfy the requirements of Definition~\ref{def:perfectcodes}.


\vspace{2mm}

The approximate universality of an ST code over the MIMO channel was studied in~\cite{tv06}. This property refers to an ST code being optimal in terms of DMT regardless of the fading statistics of $\bH$. A sufficient and necessary condition for an ST code to be approximately universal was derived in~\cite{tv06}. This condition is closely related to the nonvanishing determinant criterion and is satisfied by perfect linear dispersion ST codes. The next Theorem is an extension of~\cite[Theorem 3.1]{tv06}. The notation $[x]^+\triangleq\max(x,0)$ is used.

\vspace{2mm}

\begin{theorem}
\label{thm:normbound}
Let $\Cst_{\infty}$ be an $M\times M$ perfect linear dispersion ST code over a $\text{QAM}(\infty)$ constellation with $\delta_{\text{min}}(\Cst_{\infty})=\inf_{\mathbf{0}\neq\bX\in\Cst_{\infty}}|\det(\bX)|^2>0$, and let $\Cst$ be its subcode over a $\text{QAM}(L)$ constellation. Then, for all channel matrices $\bH$ with corresponding WI mutual information \mbox{$\Cwi=\log\det(\bI+\Tsnr\bH^{\dagger}\bH)$}, $M$ transmit antennas and an arbitrary number of receive antennas and all $\mathbf{0}\neq\bX\in\Cst$
\begin{align}
\Tsnr\|\bH\bX\|_F^2\geq \left[\delta_{\text{min}}(\Cst_{\infty})^{\frac{1}{M}}2^{\frac{\Cwi}{M}}-2M^2L^2\right]^+.\nonumber
\end{align}
\end{theorem}

\vspace{2mm}

\begin{proof}
The proof closely follows that of~\cite[Theorem 3.1]{tv06}, and is given in Appendix~\ref{app:proofnormbound}.
\end{proof}

\vspace{2mm}

Let $\MH=\bI_M\otimes\bH$, as in~\eqref{hconc}.
The next simple corollary of Theorem~\ref{thm:normbound} will be used in Section~\ref{sec:main} to prove the main result of this paper.

\vspace{2mm}

\begin{corollary}
\label{cor:QAM}
Let $\bP\in\CC^{M^2\times M^2}$ be a generating matrix of a perfect \mbox{$M\times M$} QAM based linear dispersion ST code $\Cst_\infty$ with \mbox{$\delta_{\text{min}}(\Cst_{\infty})=\inf_{\mathbf{0}\neq\bX\in\Cst_{\infty}}|\det(\bX)|^2>0$}. Then, for all channel matrices $\bH$ with corresponding WI mutual information \mbox{$\Cwi=\log\det(\bI+\Tsnr\bH^{\dagger}\bH)$}, $M$ transmit antennas and an arbitrary number of receive antennas
\begin{align}
\Tsnr d^2_{\min}(\MH\bP,L)\geq \left[\delta_{\text{min}}(\Cst_{\infty})^{\frac{1}{M}}2^{\frac{\Cwi}{M}}-2M^2L^2\right]^+.\nonumber
\end{align}
\end{corollary}

\begin{proof}
Consider the subcode $\Cst$ of $\Cst_\infty$, defined over a $\text{QAM}(L)$ constellation. Then, for any \mbox{$\ba\in\text{QAM}^{M^2}(L)$} there exist a code matrix $\bX\in\Cst$ such that
\begin{align}
\vect(\bX)=\bP\ba.\nonumber
\end{align}
Now,
\begin{align}
\Tsnr\|\MH\bP\ba\|^2&=\Tsnr\|\MH\vect(\bX)\|^2\nonumber\\
&=\Tsnr\|\bH\bX\|^2_F\nonumber\\
&\geq \left[\delta_{\text{min}}(\Cst_{\infty})^{\frac{1}{M}}2^{\frac{\Cwi}{M}}-2M^2L^2\right]^+,\nonumber
\end{align}
where the last inequality follows from Theorem~\ref{thm:normbound}. It follows that
\begin{align}
\Tsnr d_{\text{min}}^2(\MH\bP,L)&=\min_{\ba\in \text{QAM}^{M^2}(L)\setminus\mathbf{0}}\Tsnr\|\MH\bP\ba\|^2\nonumber\\
&\geq \left[\delta_{\text{min}}(\Cst_{\infty})^{\frac{1}{M}}2^{\frac{\Cwi}{M}}-2M^2L^2\right]^+.\nonumber
\end{align}

\end{proof}

\section{Main Result}
\label{sec:main}

The next theorem lower bounds the effective signal-to-noise ratio of precoded IF equalization, where the precoding matrix generates a perfect linear dispersion ST code. The obtained bound depends on the channel matrix $\bH$ \emph{only} through its corresponding WI mutual information.

\begin{theorem}
\label{thm:snrbound}
Consider the aggregate MIMO channel
\begin{align}
\bar{\by}&=\mathcal{H}\bP\bar{\bx}+\bar{\bz}\nonumber
\end{align}
where $\MH=\bI_M\otimes\bH\in\CC^{NM\times M^2}$, and $\bP\in\CC^{M^2\times M^2}$ is a generating matrix of a perfect \mbox{$M\times M$} QAM based linear dispersion ST code $\Cst_\infty$ with \mbox{$\delta_{\text{min}}(\Cst_{\infty})=\inf_{\mathbf{0}\neq\bX\in\Cst_{\infty}}|\det(\bX)|^2>0$}. Then, applying IF equalization to the aggregate channel yields
\begin{align}
\Tsnr_{\text{eff}}>\frac{1}{8 M^6}\delta_{\text{min}}(\Cst_{\infty})^{\frac{1}{M}}2^{\frac{\Cwi}{M}},\nonumber
\end{align}
for all channel matrices $\bH$ with corresponding WI mutual information \mbox{$\Cwi=\log\det(\bI+\Tsnr\bH^{\dagger}\bH)$}, $M$ transmit antennas and an arbitrary number of receive antennas
\end{theorem}

\vspace{2mm}

\begin{proof}
Applying Lemma~\ref{lem:snrdmin} to the aggregate $NM\times M^2$ channel  matrix $\bar{\bH}={\mathcal{H}\bP}$ gives
\begin{align}
\Tsnr_{\text{eff}}&>\frac{1}{4M^4}\min_{L=1,2,\ldots}\left(L^2+\Tsnr d_{\text{min}}^2(\bar{\bH},L)\right).\label{aggdmin}
\end{align}
Using Corollary~\ref{cor:QAM}, this is bounded by
\begin{align}
\Tsnr_{\text{eff}}&>\frac{1}{4M^4}\min_{L=1,2,\ldots}\left(L^2+\left[\delta_{\text{min}}(\Cst_{\infty})^{\frac{1}{M}}2^{\frac{\Cwi}{M}}-2M^2L^2\right]^+ \right)\nonumber\\
&\geq\frac{1}{4M^4}\min_{L=1,2,\ldots}\left(L^2+\left[\frac{\delta_{\text{min}}(\Cst_{\infty})^{\frac{1}{M}}2^{\frac{\Cwi}{M}}}{2M^2}-L^2\right]^+ \right)\nonumber\\
&\geq\frac{1}{8M^6}\delta_{\text{min}}(\Cst_{\infty})^{\frac{1}{M}}2^{\frac{\Cwi}{M}}\nonumber
\end{align}
as desired.
\end{proof}

\vspace{2mm}

The next theorem shows that precoded IF attains the compound MIMO capacity to within a constant gap.

\vspace{2mm}

\begin{theorem}
Let $\bP\in\CC^{M^2\times M^2}$ be a generating matrix of a perfect \mbox{$M\times M$} QAM based linear dispersion ST code $\Cst_\infty$ with \mbox{$\delta_{\text{min}}(\Cst_{\infty})=\inf_{\mathbf{0}\neq\bX\in\Cst_{\infty}}|\det(\bX)|^2>0$}.
For all channel matrices $\bH$ with $M$ transmit antennas and an arbitrary number of receive antennas, precoded integer-forcing with the precoding matrix $\bP$ achieves any rate satisfying
\begin{align}
R_{\text{P-IF}}<\Cwi-\Gamma\left(\delta_{\text{min}}(\Cst_{\infty}),M\right),\nonumber
\end{align}
where \mbox{$\Cwi=\log\det(\bI+\Tsnr\bH^{\dagger}\bH)$}, and
\begin{align}
\Gamma\left(\delta_{\text{min}}(\Cst_{\infty}),M\right)\triangleq \log\frac{1}{\delta_{\text{min}}(\Cst_{\infty})}+3M\log(2M^2). \label{Gammadef}
\end{align}
\label{thm:constgap}
\end{theorem}

\vspace{1mm}

\begin{proof}
In precoded IF, the matrix $\bP$ is used as a precoding matrix that transforms the original $N\times M$ MIMO channel~\eqref{channelmodel} to the aggregate $NM\times M^2$ MIMO channel
\begin{align}
\bar{\by}&=\mathcal{H}\bP\bar{\bx}+\bar{\bz}\nonumber\\
&=\bar{\bH}\bar{\bx}+\bar{\bz},\label{uncodedST1}
\end{align}
as described in Section~\ref{sec:IFST}, and then IF equalization is applied to the aggregate channel.  Assuming a ``good'' nested lattice codebook is used to encode all $2M^2$ streams transmitted over the aggregate channel, by~\eqref{Rif}, IF equalization can achieve any rate satisfying
\begin{align}
R_{\text{IF,aggregate}}<M^2\log(\Tsnr_{\text{eff}}).\nonumber
\end{align}
Using Theorem~\ref{thm:snrbound}, it follows that any rate satisfying
\begin{align}
R_{\text{IF,aggregate}}&<M^2\log\left(\frac{1}{8M^6}\delta_{\text{min}}(\Cst_{\infty})^{\frac{1}{M}}2^{\frac{\Cwi}{M}}\right)\nonumber\\
&=M\Cwi-M\log\frac{1}{\delta_{\text{min}}(\Cst_{\infty})}-M^2\log(8M^6)\nonumber
\end{align}
is achievable over the aggregate channel.

Since each channel use of the aggregate channel~\eqref{uncodedST1} corresponds to $M$ channel uses of the original channel~\eqref{channelmodel}, the communication rate should be normalized by a factor of $1/M$. Thus, \mbox{$R_{\text{P-IF}}=R_{\text{IF,aggregate}}/M$}, and the theorem follows.
\end{proof}

%

\vspace{1mm}

\begin{example}
The Golden-code~\cite{brv05} is a QAM-based perfect $2\times2$ linear dispersion space time code, with \mbox{$\delta_{\text{min}}(\Cst_{\infty})=1/5$}. Thus, for a MIMO channel with $M=2$ transmit antennas, its generating matrix $\bP\in\CC^{4\times 4}$ can be used for precoded integer-forcing. Theorem~\ref{thm:constgap} implies that with this choice of $\bP$, precoded integer-forcing achieves $\Cwi$ to within a gap of \mbox{$\Gamma\left(1/5,2\right)=20.32$} bits, which translates to a gap of $5.08$ bits per real dimension. In fact, using a slightly more careful analysis,\footnote{Namely, the product of successive minima of a lattice and its dual lattice in Theorem~\ref{thm:dualsuc} can be bounded using Proposition 3.3 from~\cite{lls90} instead of the result from~\cite{Banaszczyk93}. The bound from~\cite{lls90} involves Hermite's constant and gives better results than those obtained using~\cite{Banaszczyk93} only when very small values of $M$ are of interest.} it can be shown that, with this choice of $\bP$, precoded integer-forcing achieves $\Cwi$ to within $15.24$ bits, i.e., $3.81$ bits per real dimension.
\label{ex:specificgaps}
\end{example}

\vspace{2mm}

While the constants from Example~\ref{ex:specificgaps} may seem quite large, one has to keep in mind that this is a worst-case bound, whereas for the typical case, under common statistical assumptions such as Rayleigh fading, the gap-to-capacity obtained by precoded IF is considerably smaller, as demonstrated in Figure~\ref{fig:GapHist}.

Moreover, the recent work of Fischler \emph{et al.}~\cite{fe14} demonstrates that for channels with a special structure, the gap can be much smaller when precoded IF-SIC~\cite{oen13} is used. In particular,~\cite{fe14} studies the compound \emph{parallel} MIMO channel and finds that for channels of dimensions $2\times 2$ and $3\times 3$ precoded IF-SIC achieves at least $94\%$ and $82\%$, respectively, of the compound capacity \emph{for any value of capacity}. Theorem~\ref{thm:constgap} provides an additive bound on the gap-to-capacity, and therefore guarantees that the fraction of the compound MIMO capacity achieved by precoded IF approaches $100\%$ as the compound capacity increases. It does not, however, provide useful efficiency guarantees, i.e. multiplicative bounds, for small capacities. The results in~\cite{fe14} indicate that with a slightly more complex scheme that also incorporates successive interference cancelation, and a more limited channel model (parallel MIMO channel instead of the general MIMO channel studied here), excellent performance can be guaranteed also for low capacities.

\begin{figure}[htb]
\includegraphics[width=1 \columnwidth]{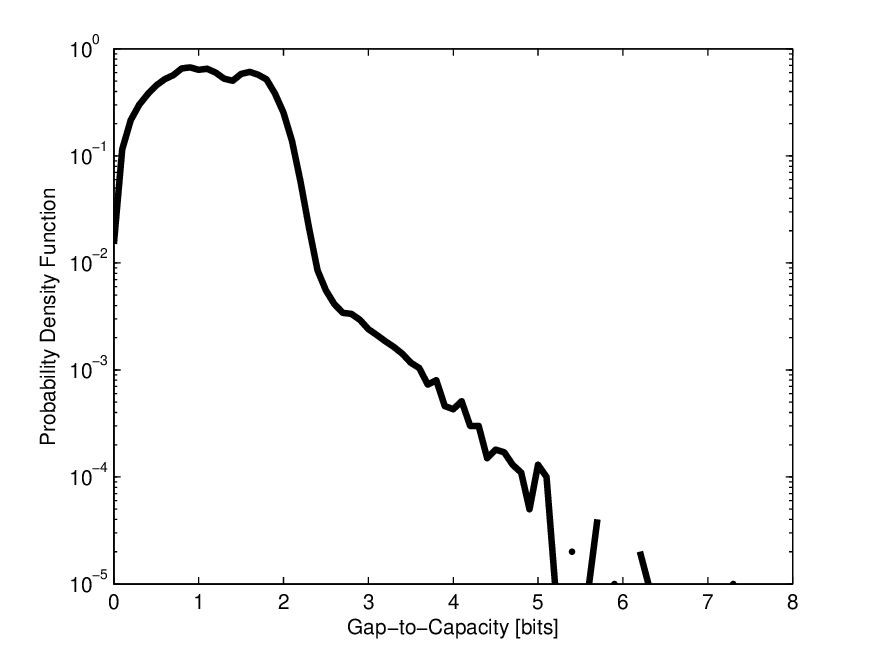}
\caption{The probability density function of the gap-to-capacity achieved by precoded IF with ``good'' nested lattices over a $2\times 2$ MIMO channel with Rayleigh fading, where after drawing $\bH$ it is scaled such that $\log\det\left|\bI+\Tsnr\bH^{\dagger}\bH\right|=30$bits. The precoded matrix that was used is the generating matrix of the Golden code. The probability that precoded IF achieves less than $90\%$ of capacity is smaller than $0.0015$ in this scenario.}
\label{fig:GapHist}
\end{figure}

\vspace{2mm}

Note that although the proof of Theorem~\ref{thm:constgap} assumed that a ``good'' nested lattice code was used, a similar result holds when a $q$-ary linear code without shaping is used. This follows from the fact that the performance of the latter is only degraded by no more than the shaping loss of $\log(2\pi e/12)$ bits per antenna w.r.t. the former. Moreover, Theorem~\ref{thm:snrbound} can also be used to obtain an upper bound on the error probability of precoded IF with uncoded transmission.

\vspace{1mm}

\begin{proposition}
\label{prop:uncoded}
For all channel matrices $\bH$ with corresponding WI mutual information \mbox{$\Cwi=\log\det(\bI+\Tsnr\bH^{\dagger}\bH)$}, $M$ transmit antennas and an arbitrary number of receive antennas, the error probability of precoded IF with uncoded transmission is bounded by
\begin{align}
P_{e,\text{P-IF-uncoded}}\leq 4M^2\exp\left\{-\frac{3}{2}2^{\frac{1}{M}\left(\Cwi-R_{\text{P-IF}}-\Gamma\left(\delta_{\text{min}}(\Cst_{\infty}),M\right)\right)} \right\},\nonumber
\end{align}
provided that the precoding matrix $\bP$ generates an \mbox{$M\times M$} perfect linear dispersion ST code $\Cst_{\infty}$ with minimum determinant $\delta_{\text{min}}(\Cst_{\infty})=\inf_{\mathbf{0}\neq\bX\in\Cst_{\infty}}|\det(\bX)|^2>0$.
\end{proposition}

\vspace{1mm}

\begin{proof}
Using~\eqref{peuncoded}, the error probability of uncoded IF equalization over the aggregate channel~\eqref{uncodedST1} is bounded by
\begin{align}
P_{e,\text{P-IF-uncoded}}\leq 4M^2\exp\left\{-\frac{3}{2}2^{\frac{1}{M^2}\left(M^2\log(\Tsnr_{\text{eff}})-MR_{\text{P-IF}}\right)} \right\},\nonumber
\end{align}
where we have used the fact that the transmission rate over the aggregate channel is $M$ times larger than the actual communication rate $R_{\text{P-IF}}$. Now, replacing $\Tsnr_{\text{eff}}$ with its bound from Theorem~\ref{thm:snrbound} establishes the proposition.
\end{proof}

\section{Application: Rateless coding for MIMO channels via precoded integer-forcing}
\label{sec:rateless}

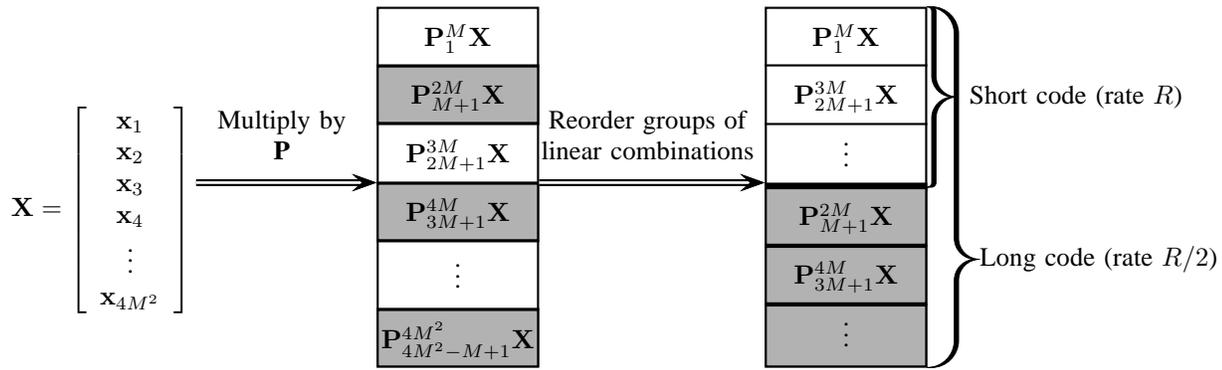
\begin{figure*}[]
\begin{center}
\psset{unit=0.6mm}
\begin{pspicture*}(0,10)(300,90)

\rput(0,0)
{
\rput(3,0){\rput(22,45){$\bX=\left[
                                                   \begin{array}{c}
                                                     \bx_1 \\
                                                     \bx_2 \\
                                                     \bx_3 \\
                                                     \bx_4 \\
                                                     \vdots \\
                                                     \bx_{4M^2} \\
                                                   \end{array}
                                                 \right]
$ }
}
\rput(45,1){
\psline[doubleline=true]{->}(0,50)(40,50)
\rput(19,60){$\begin{array}{c}
                                             $\text{Multiply by}$ \\
                                             $\textbf{P}$
                                           \end{array}$}
}
\rput(85,10){
\psframe(0,0)(36,80)
\rput(0,67){\psline(0,0)(36,0)\rput(18,6){$\bP_1^M\bX$}}
\rput(0,54){\psline(0,0)(36,0)\psframe[fillcolor=lightgrey,fillstyle=solid](0,0)(36,13)\rput(18,6){$\bP_{M+1}^{2M}\bX$}}
\rput(0,41){\psline(0,0)(36,0)\rput(18,6){$\bP_{2M+1}^{3M}\bX$}}
\rput(0,28){\psline(0,0)(36,0)\psframe[fillcolor=lightgrey,fillstyle=solid](0,0)(36,13)\rput(18,6){$\bP_{3M+1}^{4M}\bX$}}
\rput(0,13){\psline(0,0)(36,0)\rput(18,9){$\vdots$}}
\psframe[fillcolor=lightgrey,fillstyle=solid](0,0)(36,13)\rput(18,6){$\bP_{4M^2-M+1}^{4M^2}\bX$}}

\rput(121,1){
\psline[doubleline=true]{->}(0,50)(50,50)
\rput(24,60){$\begin{array}{c}
                                             $\text{Reorder groups of}$ \\
                                             $\text{linear combinations}$
                                           \end{array}$}
}
\rput(171,10){
\psframe(0,0)(36,80)
\rput(0,67){\psline(0,0)(36,0)\rput(18,6){$\bP_1^M\bX$}}
\rput(0,54){\psline(0,0)(36,0)\rput(18,6){$\bP_{2M+1}^{3M}\bX$}}
\rput(0,40){\psline[linewidth=2](0,0)(36,0)\rput(18,9){$\vdots$}}
\rput(0,27){\psline(0,0)(36,0)\psframe[fillcolor=lightgrey,fillstyle=solid](0,0)(36,13)\rput(18,6){$\bP_{M+1}^{2M}\bX$}}
\rput(0,14){\psline(0,0)(36,0)\psframe[fillcolor=lightgrey,fillstyle=solid](0,0)(36,13)\rput(18,6){$\bP_{3M+1}^{4M}\bX$}}
\psframe[fillcolor=lightgrey,fillstyle=solid](0,0)(36,14)\rput(18,9){$\vdots$}
\psbrace[ref=lt,braceWidthInner=1](36,40)(36,80){}\rput(69,60){Short code (rate $R$)}
\psbrace[ref=lt,bracePos=0.3,braceWidthInner=6](36,0)(36,80){}\rput(74,24){Long code (rate $R/2$)}
}
}

\end{pspicture*}
\end{center}
\caption{An illustration of the proposed rateless code construction. $P_\ell^k$ denotes the matrix obtained by taking the $\ell$th up to $k$th rows of the matrix $\bP$.}
\label{fig:ratelessconst}
\end{figure*} 

A notable feature of precoded IF is that the scheme, as well as its performance guarantees, do not depend on the number of antennas at the receiver side. In this section, we exploit this property for developing efficient rateless codes for the MIMO channel. The rateless coding problem is another instance of a DoF-mismatch scenario, where the transmitter has to simultaneously transmit to different (virtual) users, each with a different number of receive antennas.

In an open-loop scenario, in addition to not knowing the channel gains, the transmitter may also not know the capacity of its link to the receiver. A reasonable approach, in this case, is to transmit a long codeword describing the information bits, such that if the channel is ``good'', the receiver can stop listening after a short while, whereas if it is ``bad'' a longer fraction of the codeword is needed to ensure correct decoding. Since the code's rate is not predefined, and depends on the channel condition, such an approach is referred to as \emph{rateless coding}.

A rateless code is defined as a family of codes that has the property that codewords of the higher rate codes are prefixes of those of the lower rate ones. A family of such codes is called \emph{perfect} (not to be confused with perfect linear dispersion ST codes) if each of the codes in the family is capacity-achieving.

In this section, we show how precoded IF can be used for constructing a rateless code for the MIMO channel which is a constant number of bits from perfect, i.e., each of its subcodes achieves the compound MIMO capacity to within a constant number of bits. For sake of brevity, we only illustrate the scheme through an example rather than give a full description.

Assume the channel model is the one from~\eqref{channelmodel}, and the goal is to design two codes with rates $R$, and $R/2$, where the higher rate code is a prefix of the lower rate one. It is further required that for some predefined $\delta>0$ if the channel's capacity $C$ satisfies $C>R+\delta$ the high-rate (short) code can be decoded reliably, and if $C>R/2+\delta$ the low-rate (long) code can be decoded reliably.
This problem can be viewed as that of designing a code which is simultaneously good for the two channel matrices
\begin{align}
\bH_1=\left[
        \begin{array}{cc}
          \bH & \mathbf{0} \\
          \mathbf{0} & \mathbf{0} \\
        \end{array}
      \right] \ \ \ \text{and} \ \  \ \bH_2=\left[
        \begin{array}{cc}
          \bH & \mathbf{0} \\
          \mathbf{0} & \bH \\
        \end{array}
      \right],\nonumber
\end{align}
since the effective channel $\bH_2$ is obtained from twice as many channel uses as $\bH_1$, which corresponds to a code twice as long.
If $\bH\in\CC^{N\times M}$, then $\bH_1,\bH_2\in\CC^{2N\times 2M}$. In the previous section, it was shown that precoded IF can simultaneously achieve the capacity of any MIMO channel to within a constant gap. In particular, it can simultaneously achieve the capacity of $\bH_1$ and $\bH_2$ to within a constant gap.

The rateless code is therefore constructed from $4M^2$ complex streams of linear codewords (each consisting of one linear codeword in its quadrature component and one in its in-phase components). Each complex stream is of length $n$ and carries $nR/2M$ bits. These streams are then precoded using the matrix $\bP\in\CC^{4M^2\times 4M^2}$ which generates a perfect $2M\times 2M$ linear dispersion ST code. This results in a set of $4M^2$ linear combinations of the coded streams. The linear combinations are then split into $4M$ groups each containing $M$ linear combinations, such that the first group consists of the first $M$ linear combinations, the next group contains the next $M$ linear combinations, and so on. The short code consists of the odd groups of linear combinations, whereas the long code consists of both odd and even groups of linear combinations. See Figure~\ref{fig:ratelessconst} for an illustration of the code construction.

The long code is transmitted during $4Mn$ consecutive channel uses. At the receiver side, integer-forcing equalization is applied. The receiver, which knows the channel capacity, can decide whether the first $2M^2$ linear combinations, corresponding to the first $2Mn$ channel uses, suffice for correct decoding of the $4M^2$ coded streams, or all $4M^2$ linear combinations, corresponding to all $4Mn$ channel uses, are needed. Theorem~\ref{thm:constgap} implies that if the capacity is greater than \mbox{$R+\Gamma\left(\delta_{\text{min}}(\Cst_{\infty}),2M\right)$} the short code can be decoded reliably, and if it is greater than \mbox{$R/2+\Gamma\left(\delta_{\text{min}}(\Cst_{\infty}),2M\right)$} the long code can be decoded reliably.

Note that although we have only described the construction of a code that is compatible with two different rates, the aforementioned construction can be easily extended to any number of rates.

\section{Discussion and Summary}
\label{sec:conclusions}
The additive Gaussian noise MIMO channel in an open-loop scenario, where the receiver has complete channel state information whereas the transmitter only knows the white-input mutual information was considered in this paper. It was shown that using linear precoding at the transmitter in conjunction with integer-forcing equalization at the receiver suffices to approach the capacity of this compound channel to within a constant gap, depending only on the number of transmit antennas.
To the best of our knowledge, this is the first practical scheme that guarantees an additive loss w.r.t. the compound capacity. Such a performance guarantee is much stronger than DMT optimality, which is at present the common benchmark for evaluating schemes. In particular, although our results are free from any statistical assumptions, they can be interpreted to obtain performance guarantees in a MIMO fading environment. Specifically, a scheme that achieves a constant gap from capacity is DMT optimal under any fading statistics, and achieves a constant gap from the outage capacity under any fading statistics.

IF equalization uses coded streams, and is therefore usually less suitable for fast fading environments. Nevertheless, we have also developed new upper bounds on an uncoded version of IF equalization, which is more adequate for fast fading. We note that while uncoded IF equalization is quite similar to lattice reduction aided decoding, to the best of our knowledge, the performance of the latter was never analyzed at such a fine scale.

Another appealing feature of the described scheme, inherited from the properties of its underlying perfect ST codes, is that it is independent of the number of receive antennas, and the performance guarantees obtained in this paper do not depend on the number of receive antennas as well. Hence, the scheme is not sensitive to a degrees-of-freedom mismatch.

The compound channel studied in this paper includes all channel matrices with the same white-input mutual information. In certain scenarios, such as multicasting the same message to a finite set of users whose channel matrices are known at the transmitter, it makes sense to consider compound channels with a relatively small number of users. Recent work~\cite{de14} demonstrates that precoded IF-SIC performs remarkably well in such scenarios and achieves a large fraction of the compound capacity, even at small SNRs, under reasonable statistical assumptions on the distribution of the channel matrices.

\section*{Acknowledgment}
Helpful discussions with Yair Yona, Elad Domanovitz, Barak Weiss and Bobak Nazer are greatly acknowledged.
\begin{appendices}

\section{Proof of Lemma~\ref{lem:IFuncoded}}
\label{app:uncoded}

The output of the $k$th sub-channel with uncoded transmission is
\begin{align}
\tilde{y}_k&=[v_k+z_{\text{eff},k}]\mod \gamma q\ZZ\nonumber,
\end{align}
where $v_k\in\gamma\ZZ$. The estimate $\hat{v}_k$ is generated by applying a simple slicer (nearest-neighbor quantizer w.r.t. $\gamma\ZZ$) to $\tilde{y}_k$, followed by \hspace{-1.5mm}$\mod \gamma q\ZZ$ reduction.
The detection error probability at the $k$th sub-channel is upper bounded by
\begin{align}
P_{e,k}&\triangleq\Pr\left(\hat{v}_k\neq v_k\right)\nonumber\\
&\leq\Pr\left(\left|z_{\text{eff},k}\right|\geq\frac{\gamma}{2}\right).\nonumber
\end{align}
In order to bound $P_{e,k}$, a simple lemma, which is based on~\cite[Theorem 7]{fsk11} is needed.

\vspace{2mm}

\begin{lemma}
\label{lem:mixture}
Consider the random variable
\begin{align}
z_{\text{eff}}=\sum_{\ell=1}^{L} \alpha_\ell z_\ell +\sum_{k=1}^K \beta_k d_k\nonumber
\end{align}
where $\left\{z_\ell\right\}_{\ell=1}^L$ are i.i.d. Gaussian random variables with zero mean and some variance $\sigma^2_z$ and $\left\{d_k\right\}_{k=1}^K$ are i.i.d. random variables, statistically independent of $\left\{z_\ell\right\}_{\ell=1}^L$, uniformly distributed over the interval $[-\rho/2,\rho/2)$ for some $\rho>0$. Let $\sigma^2_{\text{eff}}\triangleq\mathbb{E}(z^2_{\text{eff}})$. Then
\begin{align}
\Pr(z_{\text{eff}}>\tau)=\Pr(z_{\text{eff}}<-\tau)\leq\exp\left\{-\frac{\tau^2}{2\sigma^2_{\text{eff}}}\right\}.\nonumber
\end{align}
\end{lemma}

\vspace{2mm}

\begin{proof}
The probability density function of $z_{\text{eff}}$ is symmetric around zero and hence
\begin{align}
\Pr(z_{\text{eff}}\geq\tau)=\Pr(z_{\text{eff}}\leq-\tau).\nonumber
\end{align}
Applying Chernoff's bound gives (for $s>0$)
\begin{align}
\Pr(z_{\text{eff}}\geq\tau)&\leq e^{-s\tau}\mathbb{E}\left(e^{s z_{\text{eff}}}\right)\nonumber\\
&=e^{-s\tau}\mathbb{E}\left(e^{s \left(\sum_{\ell=1}^L\alpha_\ell z_l+\sum_{k=1}^K\beta_k d_k\right)}\right)\nonumber\\
&=e^{-s\tau}\prod_{\ell=1}^L\mathbb{E}\left(e^{s\alpha_\ell z_l}\right)\prod_{k=1}^K\mathbb{E}\left(e^{s\beta_k d_k}\right).\nonumber
\end{align}
Using the well-known expressions for the moment generating functions of Gaussian and uniform random variables gives
\begin{align}
\mathbb{E}\left(e^{s\alpha_\ell z_l}\right)&=e^{\frac{1}{2}s^2\alpha^2_\ell\sigma^2_z},\nonumber\\
\mathbb{E}\left(e^{s\beta_k d_k}\right)&=\frac{\sinh(s\beta_k\rho/2)}{s\beta_k\rho/2}\leq e^{\frac{1}{2}\frac{s^2\beta_k^2\rho^2}{12}},\nonumber
\end{align}
where the last inequality follows from \mbox{$\sinh(x)/x\leq \exp\{x^2/6\}$} (which can be obtained by simple Taylor expansion)~\cite{fsk11}. It follows that
\begin{align}
\Pr(z_{\text{eff}}\geq\tau)&\leq e^{-s\tau}e^{\frac{s^2}{2}\left(\sum_{\ell=1}^L \alpha^2_\ell\sigma^2_z+\sum_{k=1}^K\beta^2_k\frac{\rho^2}{12} \right)}\nonumber\\
&=e^{-s\tau+\frac{1}{2}s^2\sigma^2_{\text{eff}}}.\label{sbound}
\end{align}
Setting $s=\tau/\sigma^2_{\text{eff}}$ gives the desired result.
\end{proof}

\vspace{2mm}

Now, using Lemma~\ref{lem:mixture}, the probability of detection error at the $k$th sub-channel can be bounded as
\begin{align}
P_{e,k}
&\leq\Pr\left(\left|z_{\text{eff},k}\right|\geq\frac{\gamma}{2} \right)\nonumber\\
&\leq 2\exp\left\{-\frac{\gamma^2}{8\sigma^2_{\text{eff},k}} \right\}\nonumber\\
&= 2\exp\left\{-\frac{12\Tsnr}{8 q^2\sigma^2_{\text{eff},k}} \right\}\nonumber\\
&= 2\exp\left\{-\frac{3}{2}\frac{1}{q^2}\Tsnr_{\text{eff},k} \right\},\nonumber
\end{align}
where the definition of $\Tsnr_{\text{eff},k}$ was used in the last equality. Using the fact that $q=2^R$ and that \mbox{$\Tsnr_{\text{eff},k}\geq\Tsnr_{\text{eff}}$} for all $k=1,\ldots,2M$, the detection error probability at each of the $2M$ sub-channels can be further bounded as
\begin{align}
P_e\leq 2\exp\left\{-\frac{3}{2}2^{2\left(\frac{1}{2}\log(\Tsnr_{\text{eff}})-R\right)} \right\}.\nonumber
\end{align}
Since the IF equalizer makes an error only if a detection error occurred in at least one of the $2M$ sub-channels, and since the total transmission rate is \mbox{$R_{\text{IF}}=2MR$}, the total error probability of the IF equalizer with uncoded transmission rate $R_{\text{IF}}$ is bounded by
\begin{align}
P_{e,\text{IF-uncoded}}&\leq 4M\exp\left\{-\frac{3}{2}2^{2\left(\frac{1}{2}\log(\Tsnr_{\text{eff}})-\frac{R_{\text{IF}}}{2M}\right)} \right\}\nonumber\\
&=4M\exp\left\{-\frac{3}{2}2^{\frac{1}{M}\left(M\log(\Tsnr_{\text{eff}})-R_{\text{IF}}\right)} \right\}.\nonumber
\end{align}

\section{Proof of Lemma~\ref{lem:IF_DoF}}
\label{app:DoFlemma}
Let $\psi:\RR^+\rightarrow\RR^+$ be a real positive decreasing function with $\psi(r)\rightarrow 0$ as $r\rightarrow\infty$, let $\mathbb{I}^{N\times M}\triangleq [-1/2,1/2)^{N\times M}$ be the set of all matrices of dimensions $N\times M$ with all entries taken from the interval $[-1/2,1/2)$, and define the set
\begin{align}
W_0(M,N,\psi)\triangleq\bigg\{\bH\in\mathbb{I}^{N\times M} \ &: \ ||\bH\ba\|_{\infty}\leq \psi(\|\ba\|_{\infty}) \nonumber\\
&\ \text{for i.m. } \ba\in\ZZ^M\setminus\mathbf{0} \bigg\},
\end{align}
where $\|\bx\|_{\infty}\triangleq\max_{i}|x_i|$ is the infinity norm, and i.m. means infinitely many.
The next result from~\cite[Corollary 2]{hussain2009metrical} shows that $W_0(M,N,\psi)$ has either zero Lebesgue measure or full Lebesgue measure, depending on the choice of the function $\psi$.

\vspace{2mm}

\begin{theorem}{\cite[Corollary 2]{hussain2009metrical}}
Let $\psi:\RR^+\rightarrow\RR^+$ be a real positive decreasing function with $\psi(r)\rightarrow 0$ as $r\rightarrow\infty$. For $M>N$, if the series
\begin{align}
\sum_{r=1}^{\infty}\psi^N(r)r^{M-N-1}\nonumber
\end{align}
converges then the set $W_0(M,N,\psi)$ has zero Lebesgue measure, and if it diverges the set $W_0(M,N,\psi)$ has full Lebesgue measure.
\label{Thm:mumtaz}
\end{theorem}

\vspace{2mm}

For the choice $\psi(r)=r^{-(\frac{M+\epsilon}{N}-1)}$, $\epsilon>0$, the sum from Theorem~\ref{Thm:mumtaz} converges. This, implies that for $M>N$ the set
\begin{align}
\tilde{W}_0(M,N)\triangleq\bigg\{\bH\in\mathbb{I}^{N\times M} \ &: \ ||\bH\ba\|_{\infty}\leq \|\ba\|_{\infty}^{-(\frac{M+\epsilon}{N}-1)} \nonumber\\
&\ \text{for i.m. } \ba\in\ZZ^M\setminus\mathbf{0} \bigg\}
\end{align}
has zero measure.
Define the set
\begin{align}
\mathcal{S}_0(M,N)\triangleq\bigg\{\bH\in\mathbb{I}^{N\times M} \ &: \ \|\bH\ba\|^2\leq \|\ba\|_{\infty}^{-2(\frac{M+\epsilon}{N}-1)} \nonumber\\
&\ \text{for i.m. } \ba\in\ZZ^M\setminus\mathbf{0} \bigg\}\nonumber,
\end{align}
and note that $\mathcal{S}_0(M,N)\subseteq \tilde{W}_0(M,N)$, as $\|\bH\ba\|_{\infty}^2\leq\|\bH\ba\|^2$.
The next Corollary is straightforward.

\vspace{2mm}

\begin{corollary}
For $M>N$ and any $\epsilon>0$, the set
\begin{align}
\mathcal{H}_0(M,N)
\triangleq\bigg\{\bH\in\mathbb{I}^{N\times M} \ &: \ \tilde{d}^2_{\text{min}}(\bH,L)\leq  L^{-2(\frac{M+\epsilon}{N}-1)} \nonumber\\
&\text{for i.m. } L\in\mathbb{N} \bigg\}.\nonumber
\end{align}
has zero Lebesgue measure.
\end{corollary}

\vspace{2mm}

\begin{proof}
By the definition of $\tilde{d}_{\text{min}}(\bH,L)$, the sets $\mathcal{S}_0(M,N)$ and $\mathcal{H}_0(M,N)$ are equal. The corollary then follows from the fact that $\tilde{W}_0(M,N)$ has zero measure and that $\mathcal{S}_0(M,N)\subseteq \tilde{W}_0(M,N)$.
\end{proof}

\vspace{2mm}

Let $\bar{\mathcal{H}}_0(M,N)=\mathbb{I}^{N\times M}\setminus\mathcal{H}_0(M,N)$ be the complement set of $\mathcal{H}_0(M,N)$, and note that $\bar{\mathcal{H}}_0(M,N)$ has full Lebesgue measure. For any $\bH\in\bar{\mathcal{H}}_0(M,N)$ there exist a positive integer $L^{*}(\bH)$ such that the inequality
\begin{align}
\tilde{d}^2_{\text{min}}(\bH,L)>  L^{-2(\frac{M+\epsilon}{N}-1)}
\end{align}
holds for any integer $L>L^{*}(\bH)$. It follows that
\begin{align}
\min_{L>L^*(\bH)} & \left(L^2+\Tsnr\tilde{d}^2_{\text{min}}(\bH,L)\right) \nonumber\\
&>\min_{L>L^*(\bH)} \left(L^2+ \Tsnr L^{-2(\frac{M+\epsilon}{N}-1)}\right)\nonumber\\
&>\min_{L>L^*(\bH)} \max\left(L^2, \Tsnr L^{-2(\frac{M+\epsilon}{N}-1)}\right)\nonumber\\
&>\min_{L>0} \max\left(L^2, \Tsnr L^{-2(\frac{M+\epsilon}{N}-1)}\right).\label{minmax}
\end{align}
Since $L^2$ is increasing in $L$ and $ \Tsnr L^{-2(\frac{M+\epsilon}{N}-1)}$ is decreasing in $L$, the minimum in~\eqref{minmax} is attained when $L^2= \Tsnr L^{-2(\frac{M+\epsilon}{N}-1)}$, which occurs for
\begin{align}
L^2=\Tsnr^{\frac{N}{M+\epsilon}}.\nonumber
\end{align}
This implies that
\begin{align}
\min_{L>L^*(\bH)} & L^2+\Tsnr\tilde{d}^2_{\text{min}}(\bH,L)>\Tsnr^{\frac{N}{M+\epsilon}}.\nonumber
\end{align}
On the other hand, for any $\bH\in\bar{\mathcal{H}}_0(M,N)$ we can find a constant $c(\bH)>0$ such that
\begin{align}
\min_{L\leq L^*(\bH)}\tilde{d}_{\text{min}}(\bH,L)>c(\bH).\nonumber
\end{align}
This follows from the fact that if there exist an integer vector $\ba\in\ZZ^M\setminus \mathbf{0}$ for which $\|\bH\ba\|^2=0$, then there are infinitely many such vectors, which contradicts the fact that $\bH\in\bar{\mathcal{H}}_0(M,N)$. Thus, for any $\bH\in\bar{\mathcal{H}}_0(M,N)$ we have
\begin{align}
\Tsnr_{\text{eff}}&>\frac{1}{M^2}\min_{L=1,2,\cdots} \left(L^2+\Tsnr\tilde{d}^2_{\text{min}}(\bH,L)\right)\nonumber\\
&= \frac{1}{M^2}\min\bigg(\min_{L\leq L^*(\bH)}\left(L^2+\tilde{d}^2_{\text{min}}(\bH,L)\right),\nonumber\\
& \ \ \ \ \ \ \ \  \ \ \ \ \  \  \ \ \ \min_{L> L^*(\bH)}\left(L^2+\tilde{d}^2_{\text{min}}(\bH,L)\right)\bigg)\nonumber\\
&>\frac{1}{M^2}\min\left(c(\bH)\Tsnr,\Tsnr^{\frac{N}{M+\epsilon}}\right).
\label{compBound}
\end{align}
Taking the limit of $\Tsnr\rightarrow\infty$ we see that
\begin{align}
\lim_{\Tsnr\rightarrow\infty}\frac{\nicefrac{1}{2}\log(\Tsnr_{\text{eff}})}{\nicefrac{1}{2}\log(\Tsnr)}\geq\frac{N}{M+\epsilon},\label{DoFMN}
\end{align}
for any $\bH\in\bar{\mathcal{H}}_0(M,N)$ and $M>N$. Now, taking $\epsilon\rightarrow 0$ we see that for any $\bH\in\bar{\mathcal{H}}_0(M,N)$ and $M>N$ the IF scheme achieves $N$ degrees of freedom. Since $\bH\in\bar{\mathcal{H}}_0(M,N)$ has full Lebesgue measure, the IF scheme achieves $N$ degrees of freedom for almost every $\bH\in\mathbb{I}^{N\times M}$. To see why this is also true for almost every $\bH\in\RR^{N\times M}$, note that if $\bH\notin\mathbb{I}^{N\times M}$, then we can scale it by a scalar $\rho\leq 1$ such that $\rho\bH\in\mathbb{I}^{N\times M}$. But since $\tilde{d}_{\text{min}}(\bH,L)\geq \tilde{d}_{\text{min}}(\rho\bH,L)$, this will only decrease $\Tsnr_{\text{eff}}$. Thus, we conclude that the IF scheme achieves $N$ degrees of freedom for almost every $\bH\in\RR^{N\times M}$, which establishes the lemma for $M>N$.

The case $N\geq M$ is much easier. For any matrix $\bH\in\RR^{N\times M}$ we denote the smallest singular value by $\sigma_M(\bH)$. Standard linear algebra gives
\begin{align}
\|\bH\ba\|^2\geq\sigma_M^2(\bH)\|\ba\|^2.\nonumber
\end{align}
Since $\|\ba\|\geq 1$ for all $\ba\in\text{PAM}^M(L)\setminus\mathbf{0}$, we have
\begin{align}
\Tsnr_{\text{eff}}&>\frac{1}{M^2}\min_{L=1,2,\cdots} \left(L^2+\Tsnr\tilde{d}^2_{\text{min}}(\bH,L)\right)\nonumber\\
&>\frac{\sigma^2_M(\bH)}{M^2}\Tsnr.\label{mineigbound}
\end{align}
For $N\geq M$, the set of matrices $\bH\in\RR^{N\times M}$ for which $\sigma^2_M(\bH)>0$ has full Lebesgue measure. Applying~\eqref{mineigbound} gives
\begin{align}
\lim_{\Tsnr\rightarrow\infty}\frac{\nicefrac{1}{2}\log(\Tsnr_{\text{eff}})}{\nicefrac{1}{2}\log(\Tsnr)}\geq 1,\label{DOFNM}
\end{align}
for almost every $\bH\in\RR^{N\times M}$ when $N\geq M$. Combining~\eqref{DoFMN} and~\eqref{DOFNM}, we get that
\begin{align}
\lim_{\Tsnr\rightarrow\infty}\frac{R_{\text{IF}}(\Tsnr)}{\nicefrac{1}{2}\log(\Tsnr)}\geq\min(M,N).\label{DoFLB}
\end{align}
It is well-known (see e.g.,~\cite{tseviswanath}) that for all $\bH\in\RR^{N\times M}$, the number of DoF offered by the channel is not greater $\min(M,N)$, regardless of the coding scheme which is used. Combining this with~\eqref{DoFLB} gives
\begin{align}
\lim_{\Tsnr\rightarrow\infty}\frac{R_{\text{IF}}(\Tsnr)}{\nicefrac{1}{2}\log(\Tsnr)}=\min(M,N).\label{DoFLB}
\end{align}
for almost every $\bH\in\RR^{N\times M}$, as desired.

\section{Proof of Theorem~\ref{thm:normbound}}
\label{app:proofnormbound}

Consider some arbitrary $\mathbf{0}\neq\bX\in\Cst$ and let
\begin{align}
\bH=\bU_1\mathbf{\Psi}\bV_1^{\dagger} \ \text{and} \ \bX=\bU_2\bm{\Lambda}\bV_2^{\dagger}\nonumber
\end{align}
be the singular value decompositions (SVD) of $\bH$ and $\bX$, respectively. With this notation
\begin{align}
\Tsnr\|\bH\bX\|_F^2=\Tsnr\|\mathbf{\Psi}\bV_1^{\dagger}\bU_2\bm{\Lambda}\|_F^2.\label{normtmp}
\end{align}
Suppose without loss of generality that the (absolute) singular values are ordered by increasing value in $\bm{\Lambda}$ and by decreasing value in $\mathbf{\Psi}$:
\begin{align}
\bm{\Lambda}&=\diag\{\lambda_1,\ldots,\lambda_{M}\},\nonumber\\
\mathbf{\Psi}&=\diag\{\psi_1,\ldots,\psi_{m_n},0,\cdots,0\},\nonumber
\end{align}
where $m_n\triangleq\min\{M,N\}$.
In order to establish the desired result one has to find the channel $\bH$ with corresponding WI mutual information $\Cwi$ that minimizes~\eqref{normtmp}. The rotation matrix $\bV_1$ that minimizes~\eqref{normtmp} is $\bV_1=\bU_2$ which aligns the weaker singular values of the channel matrix with the stronger singular values of the code matrix~\cite{kw03}. Thus, the problem of finding the worst channel matrix $\bH$ w.r.t. the codeword $\bX$ reduces to the optimization problem
\begin{align}
&\min_{\psi_1,\ldots,\psi_{m_n}}\Tsnr\sum_{m=1}^{m_n}|\psi_m|^2|\lambda_m|^2\nonumber\\
&\text{subject to} \sum_{m=1}^{m_n}\log(1+|\psi_m|^2\Tsnr)=\Cwi.\label{nmopt}
\end{align}
A lower bound on the solution of the minimization problem~\eqref{nmopt} can be obtained by replacing $m_n$ with $M\geq m_n$, which increases (or does not change) the optimization space and results in
\begin{align}
&\min_{\psi_1,\ldots,\psi_{M}}\Tsnr\sum_{m=1}^{M}|\psi_m|^2|\lambda_m|^2\nonumber\\
&\text{subject to} \sum_{m=1}^{M}\log(1+|\psi_m|^2\Tsnr)=\Cwi.\label{optproblem}
\end{align}
The solution to~\eqref{optproblem} is given by standard water-filling~\cite{tv06}
\begin{align}
\Tsnr\|\bH\bX\|_F^2\geq\sum_{m=1}^M\left[\frac{1}{\lambda}-|\lambda_m|^2 \right]^+,\label{normlambda}
\end{align}
where $\lambda$ satisfies
\begin{align}
\sum_{m=1}^M\left[\log\left(\frac{1}{\lambda|\lambda_m|^2}\right) \right]^+=\Cwi.\label{lambdacond}
\end{align}
Without loss of generality we may assume that \mbox{$2M^2L^2\leq \delta_{\text{min}}(\Cst_{\infty})^{\frac{1}{M}}2^{\frac{\Cwi}{M}}$} as otherwise the theorem is trivial. With this assumption, we next show that the $[\cdot]^+$ operation in~\eqref{lambdacond} is not needed, and hence its solution is given by
\begin{align}
\frac{1}{\lambda}=|\lambda_1\cdots\lambda_M|^{\frac{2}{M}}2^{\frac{\Cwi}{M}}.\label{lambdasol}
\end{align}
To see this, one has to show that with $1/\lambda$ as above the inequality $1/\lambda\geq |\lambda_m|^2$ holds for all \mbox{$m=1,\cdots,M$}. First recall that $\bX$ is a codeword from a perfect linear dispersion ST code over an $\text{QAM}(L)$ constellation. Let $\bP$ be the generating matrix of the code $\Cst$.
Thus, $\vect(\bX)=\bP\bs$ for some vector $\bs$ whose $M^2$ components all belong to the $\text{QAM}(L)$ constellation. This implies that
\begin{align}
\sum_{m=1}^M|\lambda_m|^2&=\|\bX\|_F^2\nonumber\\
&=\|\vect(\bX)\|^2\nonumber\\
&=\|\bP\bs\|^2\nonumber\\
&=\|\bs\|^2\label{unitaryP}\\
&\leq 2M^2L^2,\label{lambdaeq}
\end{align}
where~\eqref{unitaryP} follows from the fact that $\bP$ is unitary. In particular,~\eqref{lambdaeq} implies that
\begin{align}
|\lambda_m|^2\leq 2M^2L^2\nonumber
\end{align}
for all $m=1,\ldots,M$.  Since by definition
\begin{align}
|\lambda_1\cdots\lambda_M|^2=|\det(\bX)|^2\geq \delta_{\text{min}}(\Cst_{\infty}),\nonumber
\end{align}
we have for all $m=1,\ldots,M$
\begin{align}
|\lambda_m|^2&\leq 2M^2L^2\nonumber\\
&\leq \delta_{\text{min}}(\Cst_{\infty})^{\frac{1}{M}}2^{\frac{\Cwi}{M}}\nonumber\\
&\leq |\lambda_1\cdots\lambda_M|^{\frac{2}{M}}2^{\frac{\Cwi}{M}}\nonumber\\
&=\frac{1}{\lambda}.\nonumber
\end{align}
Thus,~\eqref{lambdasol} indeed solves~\eqref{lambdacond}.

Substituting~\eqref{lambdasol} into~\eqref{normlambda} gives
\begin{align}
\Tsnr\|\bH\bX\|_F^2&\geq \left[M|\lambda_1\cdots\lambda_M|^{\frac{2}{M}}2^{\frac{\Cwi}{M}}-\sum_{m=1}^M|\lambda_m|^2\right]^+\nonumber\\
&\geq \left[M\delta_{\text{min}}(\Cst_{\infty})^{\frac{1}{M}}2^{\frac{\Cwi}{M}}-2M^2L^2\right]^+\nonumber\\
&\geq \left[\delta_{\text{min}}(\Cst_{\infty})^{\frac{1}{M}}2^{\frac{\Cwi}{M}}-2M^2L^2\right]^+\nonumber
\end{align}
as desired.

\end{appendices}

\bibliographystyle{IEEEtran}
\bibliography{IF_Constant_Gap_bib}

\begin{thebibliography}{10}
\providecommand{\url}[1]{#1}
\csname url@samestyle\endcsname
\providecommand{\newblock}{\relax}
\providecommand{\bibinfo}[2]{#2}
\providecommand{\BIBentrySTDinterwordspacing}{\spaceskip=0pt\relax}
\providecommand{\BIBentryALTinterwordstretchfactor}{4}
\providecommand{\BIBentryALTinterwordspacing}{\spaceskip=\fontdimen2\font plus
\BIBentryALTinterwordstretchfactor\fontdimen3\font minus
  \fontdimen4\font\relax}
\providecommand{\BIBforeignlanguage}[2]{{%
\expandafter\ifx\csname l@#1\endcsname\relax
\typeout{** WARNING: IEEEtran.bst: No hyphenation pattern has been}%
\typeout{** loaded for the language `#1'. Using the pattern for}%
\typeout{** the default language instead.}%
\else
\language=\csname l@#1\endcsname
\fi
#2}}
\providecommand{\BIBdecl}{\relax}
\BIBdecl

\bibitem{foschini96}
G.~J. Foschini, ``Layered space-time architecture for wireless communication in
  a fading environment when using multi-element antennas,'' \emph{Bell Labs
  Technical Journal}, vol.~1, no.~2, pp. 41--59, Summer 1996.

\bibitem{fmg98}
G.~Foschini and M.~Gans, ``On limits of wireless communications in a fading
  environment when using multiple antennas,'' \emph{Wireless Personal
  Communications}, March 1998.

\bibitem{telatar99}
E.~Telatar, ``Capacity of multi-antenna \textsc{G}aussian channels,''
  \emph{European Transactions on Telecommunications}, vol.~10, no.~6, pp.
  585--595, November - December 1999.

\bibitem{tseviswanath}
D.~Tse and P.~Viswanath, \emph{Fundamentals of Wireless Communication}.\hskip
  1em plus 0.5em minus 0.4em\relax Cambridge: Cambridge University Press, 2005.

\bibitem{zt03}
L.~Zheng and D.~Tse, ``Diversity and multiplexing: a fundamental tradeoff in
  multiple-antenna channels,'' \emph{IEEE Transactions on Information Theory},
  vol.~49, no.~5, pp. 1073--1096, May 2003.

\bibitem{tv06}
S.~Tavildar and P.~Viswanath, ``Approximately universal codes over slow-fading
  channels,'' \emph{IEEE Transactions on Information Theory}, vol.~52, no.~7,
  pp. 3233--3258, July 2006.

\bibitem{ekpkh06}
P.~Elia, K.~Kumar, S.~Pawar, P.~Kumar, and H.-F. Lu, ``Explicit space-time
  codes achieving the diversity-multiplexing gain tradeoff,'' \emph{IEEE
  Transactions on Information Theory}, vol.~52, no.~9, pp. 3869--3884, Sept.
  2006.

\bibitem{orbv06}
F.~Oggier, G.~Rekaya, J.-C. Belfiore, and E.~Viterbo, ``Perfect space-time
  block codes,'' \emph{IEEE Transactions on Information Theory}, vol.~52,
  no.~9, pp. 3885--3902, Sept. 2006.

\bibitem{esk07}
P.~Elia, B.~Sethuraman, and P.~Vijay~Kumar, ``Perfect space-time codes for any
  number of antennas,'' \emph{IEEE Transactions on Information Theory},
  vol.~53, no.~11, pp. 3853--3868, Nov. 2007.

\bibitem{yw03}
H.~Yao and G.~W. Wornell, ``Achieving the full {MIMO} diversity-multiplexing
  frontier with rotation-based space-time codes,'' in \emph{Proceedings of the
  Allerton Conference on Communications, Control, and Computing}, 2003.

\bibitem{brv05}
J.-C. Belfiore, G.~Rekaya, and E.~Viterbo, ``The golden code: a $2 \times 2$
  full-rate space-time code with nonvanishing determinants,'' \emph{IEEE
  Transactions on Information Theory}, vol.~51, no.~4, pp. 1432 -- 1436, Apr.
  2005.

\bibitem{zneg12IT}
J.~Zhan, B.~Nazer, U.~Erez, and M.~Gastpar, ``Integer-forcing linear
  receivers,'' \emph{IEEE Transactions on Information Theory}, to appear.

\bibitem{hc12}
S.-N. Hong and G.~Caire, ``Reverse compute and forward: A low-complexity
  architecture for downlink distributed antenna systems,'' in \emph{Proceedings
  of the IEEE International Symposium on Information Theory (ISIT 2012)},
  Cambridge, MA, July 2012, pp. 1147--1151.

\bibitem{hc13}
------, ``Structured lattice codes for $2\times 2$ mimo interference channel,''
  in \emph{Proceedings of the IEEE International Symposium on Information
  Theory (ISIT 2013)}, July 2013, pp. 2229--2233.

\bibitem{ncnc13}
V.~Ntranos, V.~Cadambe, B.~Nazer, and G.~Caire, ``Integer-forcing interference
  alignment,'' in \emph{Proceedings of the IEEE International Symposium on
  Information Theory (ISIT 2013)}, July 2013, pp. 574--578.

\bibitem{shv13}
A.~Sakzad, J.~Harshan, and E.~Viterbo, ``Integer-forcing {MIMO} linear
  receivers based on lattice reduction,'' \emph{IEEE Transactions on Wireless
  Communications}, vol.~12, no.~10, pp. 4905--4915, October 2013.

\bibitem{de12}
E.~Domanovitz and U.~Erez, ``Combining space-time block modulation with integer
  forcing receivers,'' in \emph{Proceedings of the 27th Convention of
  Electrical Electronics Engineers in Israel (IEEEI)}, Nov. 2012, pp. 1--4.

\bibitem{hussain2009metrical}
M.~Hussain and J.~Levesley, ``The metrical theory of simultaneously small
  linear forms,'' \emph{Functiones et Approximatio Commentarii Mathematici},
  vol.~48, no.~2, pp. 167--181, 2013.

\bibitem{je10}
J.~Jalde´n and P.~Elia, ``{DMT} optimality of {LR}-aided linear decoders for a
  general class of channels, lattice designs, and system models,'' \emph{IEEE
  Transactions on Information Theory}, vol.~56, no.~10, pp. 4765--4780, Oct.
  2010.

\bibitem{ecd04}
H.~E. Gamal, G.~Caire, and M.~O. Damen, ``Lattice coding and decoding achieve
  the optimal diversity-multiplexing tradeoff of \textsc{MIMO} channels,''
  \emph{IEEE Transactions on Information Theory}, vol.~50, no.~6, pp. 968--985,
  June 2004.

\bibitem{etw08}
R.~H. Etkin, D.~N.~C. Tse, and H.~Wang, ``Gaussian interference channel
  capacity to within one bit,'' \emph{IEEE Transactions on Information Theory},
  vol.~54, no.~12, pp. 5534--5562, December 2008.

\bibitem{adt11}
S.~Avestimehr, S.~Diggavi, and D.~Tse, ``Wireless network information flow: A
  deterministic approach,'' \emph{IEEE Transactions on Information Theory},
  vol.~57, no.~4, pp. 1872--1905, April 2011.

\bibitem{od13}
A.~Ozgur and S.~Diggavi, ``Approximately achieving {G}aussian relay network
  capacity with lattice-based {QMF} codes,'' \emph{IEEE Transactions on
  Information Theory}, vol.~59, no.~12, pp. 8275--8294, Dec 2013.

\bibitem{tse2009s}
D.~Tse, ``It's easier to approximate,'' in \emph{Plenary talk, IEEE
  International Symposium on Information Theory (ISIT), Seoul, Korea}, 2009.

\bibitem{pzek11}
T.~Philosof, R.~Zamir, U.~Erez, and A.~J. Khisti, ``Lattice strategies for the
  dirty multiple access channel,'' \emph{IEEE Transactions on Information
  Theory}, vol.~57, no.~8, pp. 5006--5035, August 2011.

\bibitem{bpt10}
G.~Bresler, A.~Parekh, and D.~Tse, ``The approximate capacity of the
  many-to-one and one-to-many \textsc{G}aussian interference channels,''
  \emph{IEEE Transactions on Information Theory}, vol.~56, no.~9, pp.
  4566--4592, September 2010.

\bibitem{ng11IT}
B.~Nazer and M.~Gastpar, ``Compute-and-forward: Harnessing interference through
  structured codes,'' \emph{IEEE Transactions on Information Theory}, vol.~57,
  no.~10, pp. 6463--6486, Oct. 2011.

\bibitem{wnps10}
M.~P. Wilson, K.~Narayanan, H.~Pfister, and A.~Sprintson, ``Joint physical
  layer coding and network coding for bidirectional relaying,'' \emph{IEEE
  Transactions on Information Theory}, vol.~11, no.~56, pp. 5641--5654,
  November 2010.

\bibitem{ncl10}
W.~Nam, S.-Y. Chung, and Y.~H. Lee, ``Capacity of the \textsc{G}aussian two-way
  relay channel to within $1/2$ bit,'' \emph{IEEE Transactions on Information
  Theory}, vol.~56, no.~11, pp. 5488--5494, November 2010.

\bibitem{CoFTransformFull}
O.~Ordentlich, U.~Erez, and B.~Nazer, ``The approximate sum capacity of the
  symmetric {G}aussian {$K$}-user interference channel,'' \emph{IEEE
  Transactions on Information Theory}, vol.~60, no.~6, pp. 3450--3482, June
  2014.

\bibitem{ramibook}
R.~Zamir, \emph{Lattice Coding for Signals and Networks}.\hskip 1em plus 0.5em
  minus 0.4em\relax Cambridge: Cambridge University Press, 2014.

\bibitem{ez04}
U.~Erez and R.~Zamir, ``Achieving $\frac{1}{2}\log{(1 + \mbox{SNR})}$ on the
  \textsc{AWGN} channel with lattice encoding and decoding,'' \emph{IEEE
  Transactions on Information Theory}, vol.~50, no.~10, pp. 2293--2314, Oct.
  2004.

\bibitem{fsk11}
C.~Feng, D.~Silva, and F.~Kschischang, ``An algebraic approach to
  physical-layer network coding,'' \emph{IEEE Transactions on Information
  Theory}, vol.~59, no.~11, pp. 7576--7596, Nov 2013.

\bibitem{cs88}
J.~H. Conway and N.~J.~A. Sloane, \emph{Sphere Packings, Lattices and
  Groups}.\hskip 1em plus 0.5em minus 0.4em\relax New York: Springer-Verlag,
  1988.

\bibitem{loeliger97}
H.-A. Loeliger, ``Averaging bounds for lattices and linear codes,'' \emph{IEEE
  Transactions on Information Theory}, vol.~43, no.~6, pp. 1767--1773, Nov.
  1997.

\bibitem{et05}
U.~Erez and S.~{ten Brink}, ``A close-to-capacity dirty paper coding scheme,''
  \emph{IEEE Transactions on Information Theory}, vol.~51, no.~10, pp.
  3417--3432, Oct. 2005.

\bibitem{oen13}
O.~Ordentlich, U.~Erez, and B.~Nazer, ``Successive integer-forcing and its
  sum-rate optimality,'' in \emph{51st Annual Allerton Conference on
  Communication, Control, and Computing}, Oct 2013, pp. 282--292.

\bibitem{znoeg10}
J.~Zhan, B.~Nazer, O.~Ordentlich, U.~Erez, and M.~Gastpar, ``Integer-forcing
  architectures for {MIMO}: Distributed implementation and {SIC},'' in
  \emph{Proceedings of the Forty Fourth Asilomar Conference on Signals, Systems
  and Computers}, Nov. 2010, pp. 322--326.

\bibitem{ozeng11}
O.~Ordentlich, J.~Zhan, U.~Erez, B.~Nazer, and M.~Gastpar, ``Practical code
  design for compute-and-forward,'' in \emph{Proceedings of the IEEE
  International Symposium on Information Theory (ISIT 2011)}, St. Petersburg,
  Russia, July 2011, pp. 1876--1880.

\bibitem{oe12}
O.~Ordentlich and U.~Erez, ``A simple proof for the existence of ``good'' pairs
  of nested lattices,'' in \emph{Proceedings of the 27th Convention of
  Electrical Electronics Engineers in Israel (IEEEI)}, Nov. 2012, pp. 1--12.

\bibitem{yw02}
H.~Yao and G.~W. Wornell, ``Lattice-reduction-aided detectors for {MIMO}
  communication systems,'' in \emph{Proceedings of IEEE Globecom 2002}, Taipei,
  Taiwan, Nov. 2002, pp. 424--428.

\bibitem{wf03}
C.~Windpassinger and R.~Fischer, ``Low-complexity near-maximum-likelihood
  detection and precoding for {MIMO} systems using lattice reduction,'' in
  \emph{Proceedings of the Information Theory Workshop}, Apr. 2003, pp.
  345--348.

\bibitem{Banaszczyk93}
W.~Banaszczyk, ``New bounds in some transference theorems in the geometry of
  numbers,'' \emph{Mathematische Annalen}, vol. 296, no.~1, pp. 625--635, 1993.

\bibitem{cassels1997}
J.~W.~S. Cassels, \emph{An Introduction to the Geometry of Numbers}.\hskip 1em
  plus 0.5em minus 0.4em\relax Springer, 1997, vol.~99.

\bibitem{kemble2005groshev}
R.~Kemble, ``A {G}roshev theorem for small linear forms,'' \emph{Mathematika},
  vol.~52, pp. 79--85, 2005.

\bibitem{hussain2012metrical}
M.~Hussain and S.~Kristensen, ``Metrical results on systems of small linear
  forms,'' \emph{International Journal of Number Theory}, vol.~9, no.~3, pp.
  769--782, 2013.

\bibitem{hh02}
B.~Hassibi and B.~Hochwald, ``High-rate codes that are linear in space and
  time,'' \emph{IEEE Transactions on Information Theory}, vol.~48, no.~7, pp.
  1804--1824, Jul. 2002.

\bibitem{lls90}
J.~Lagarias, H.~{Lenstra Jr.}, and C.~Schnorr, ``Korkin-{Z}olotarev bases and
  successive minima of a lattice and its reciprocal lattice,''
  \emph{Combinatorica}, vol.~10, no.~4, pp. 333--348, Dec. 1990.

\bibitem{fe14}
O.~Fischler and U.~Erez, ``Performance of precoded integer-forcing for parallel
  {G}aussian channels,'' in \emph{2014 IEEE International Symposium on
  Information Theory (ISIT)}, June 2014, pp. 1732--1736.

\bibitem{de14}
E.~Domanovitz and U.~Erez, ``Performance of precoded integer-forcing for
  closed-loop mimo multicast,'' in \emph{to appear in the proceedings of the
  Information Theory Worshop (ITW)}, Nov. 2014.

\bibitem{kw03}
C.~Kose and R.~Wesel, ``Universal space-time trellis codes,'' \emph{IEEE
  Transactions on Information Theory}, vol.~49, no.~10, pp. 2717--2727, Oct.
  2003.

\end{thebibliography}

\end{document}